\documentclass[conference]{IEEEtran}
\usepackage{amsmath,amssymb,amsfonts}
\usepackage{algorithmic}
\usepackage{graphicx}
\usepackage{textcomp}
\usepackage{cuted}
\usepackage{xcolor}
\usepackage[hidelinks]{hyperref}
\usepackage{comment}
\usepackage{amsthm}
\usepackage[nobottomtitles]{titlesec}
\usepackage[sorting=none, style=nature]{biblatex} 
\bibliography{bib}

\usepackage[capitalize]{cleveref}
\usepackage{etoolbox}
\AtBeginEnvironment{appendices}{\crefalias{section}{appendix}}

\usepackage{caption}
\usepackage{subcaption}
\usepackage[stable]{footmisc}
\usepackage{physics}
\usepackage{cancel}

\newtheorem{theorem}{Theorem}[section]

\newtheorem{proposition}[theorem]{Proposition}
\newtheorem{lemma}[theorem]{Lemma}
\newtheorem{remark}[theorem]{Remark}

\newcommand{\EE}{\mathbb{E}}
\newcommand{\II}{\mathbb{I}}

\newcommand{\onesvec}{\boldsymbol{1}}

\newcommand{\Hvec}{\boldsymbol{H}}
\newcommand{\Mvec}{\boldsymbol{M}}

\newcommand{\tvec}{\boldsymbol{t}}
\newcommand{\uvec}{\boldsymbol{u}}

\newcommand{\SKR}{\mathrm{SKR}}
\newcommand{\SKF}{\mathrm{SKF}}

\hypersetup{ colorlinks=true, citecolor=blue, linkcolor=blue, filecolor=blue, urlcolor=blue }

\def\BibTeX{{\rm B\kern-.05em{\sc i\kern-.025em b}\kern-.08em
    T\kern-.1667em\lower.7ex\hbox{E}\kern-.125emX}}
\begin{document}

\title{Probabilistic Cutoffs in Homogeneous Quantum Repeater Chains
}

\author{
    \IEEEauthorblockN{Jeroen Grimbergen\IEEEauthorrefmark{1}\IEEEauthorrefmark{2}\IEEEauthorrefmark{3}\IEEEauthorrefmark{5}, Stav Haldar\IEEEauthorrefmark{4}, \'Alvaro G. I\~nesta, \IEEEauthorrefmark{1}\IEEEauthorrefmark{2}\IEEEauthorrefmark{3} and Stephanie Wehner\IEEEauthorrefmark{1}\IEEEauthorrefmark{2}\IEEEauthorrefmark{3}}
    \IEEEauthorblockA{\IEEEauthorrefmark{1}QuTech, Delft University of Technology, Lorentzweg 1, 2628 CJ Delft, The Netherlands}
    \IEEEauthorblockA{\IEEEauthorrefmark{2}Quantum Computer Science, EEMCS, Delft University of Technology, Mekelweg 4, 2628 CD Delft, The Netherlands}
    \IEEEauthorblockA{\IEEEauthorrefmark{3}Kavli Institute of Nanoscience, Delft University of Technology, Lorentzweg 1, 2628 CJ Delft, The Netherlands}
    \IEEEauthorblockA{\IEEEauthorrefmark{4} College of Information and Computer Science, University of Massachusetts Amherst, \\140 Governors Dr, Amherst, Massachusetts 01002, USA}
    \IEEEauthorblockA{\IEEEauthorrefmark{5} Corresponding author. Email: j.grimbergen@tudelft.nl}
}

\maketitle
\thispagestyle{plain}
\pagestyle{plain}

\begin{abstract}
We study quantum repeater chains in which entangled links between neighbouring nodes are created through heralded entanglement generation and adjacent links are swapped as soon as possible. Since heralded entanglement generation attempts succeed only probabilistically, some links will have to be stored in quantum memories at the nodes of the chain while waiting for adjacent links to be generated. The fidelity of these stored links decreases with time due to decoherence, and if they are stored for too long then this can lead to low end-to-end fidelity. Previous work has shown that the end-to-end fidelity can be improved by deterministically discarding links when their ages exceed some cutoff value. Such deterministic cutoff policies provide strict control of the fidelity of all links, but they come at the expense of having to track link ages. In this work, we introduce a probabilistic cutoff policy that does not require tracking link ages, at the cost of abandoning strict control of the fidelity. We benchmark this new probabilistic cutoff policy against a deterministic cutoff policy. We compare the policies in terms of the end-to-end rate and fidelity, and the secret-key rate. We find that even though the probabilistic cutoff policy keeps track of less state, it can provide secret-key rates of the same order of magnitude as the deterministic cutoff policy in chains with few nodes or high elementary link generation probabilities. Moreover, we identify a scenario in which the probabilistic cutoff policy can deliver end-to-end links that are required to have some minimum threshold fidelity at a higher rate than the deterministic cutoff policy.
\end{abstract}


\section{Introduction}\label{section: introduction}
A global quantum network could realize information processing tasks that are impossible with classical networks \cite{wehner2018quantum,kimble2008quantum}. Known examples of such tasks include provably secure communication \cite{bennett2014quantum, ekert1991quantum,bennett1992quantum,shor2000simple}, blind quantum computation \cite{broadbent2009universal} and distributed quantum computation \cite{gottesman1999demonstrating,cirac1999distributed,preskill1999plug}. Quantum networks could also be used for quantum sensing applications such as extending the baseline of telescopes \cite{gottesman2012longer} or clock synchronization \cite{komar2014quantum}.

The basic ingredient in many of these applications is bipartite entanglement. It is possible to create bipartite entanglement between two physically connected network nodes using a heralded entanglement generation (HEG) protocol \cite{barrett2005efficient,cabrillo1999creation}. The physical connection is typically an optical channel. It could be realized, for example, with already deployed optical fibre \cite{stolk2024metropolitan}. However, a round of the HEG protocol succeeds only probabilistically due to channel loss, and the success probability decreases with the length of the channel. For optical fibre the decrease is exponential.

Hence, to create entanglement between two distant nodes, which we will refer to as the end nodes, previous work has proposed the use of a quantum repeater chain \cite{briegel1998quantum}. The idea is to place a chain of quantum repeater nodes between the respective end nodes such that neighbouring nodes in the chain are physically connected. A HEG protocol can then be used to generate elementary entangled links between neighbouring nodes and with entanglement swapping operations \cite{zukowski1993event,duan2001long} it is possible to combine elementary links into longer links connecting non-neighbouring nodes. The process is illustrated in \cref{fig: repeater chain schematic} for a chain of five nodes. A more detailed explanation is provided in \cref{section: model}. The reader may also refer to Refs. \cite{munro2015inside,azuma2023quantum} for reviews on quantum repeater chains.

Crucially, due to the probabilistic nature of the HEG protocol, some links must be stored in the chain while waiting for their neighbouring links to be generated. These stored links are subject to time-dependent decoherence: Their fidelity to a maximally entangled state decreases with time. Previous work has addressed this problem using deterministic cutoff policies that discard links when they have been stored in memory for too long \cite{rozpkedek2019near,collins2007multiplexed,rozpkedek2018parameter,khatri2019practical,li2021efficient,haldar2023fast,kamin2023exact,inesta2023optimal, reiss2023deep}. Such deterministic cutoff policies provide strict control of the fidelity, but their implementation requires tracking the ages of all the links in the chain and communicating these ages between the nodes.

In this work, we propose a probabilistic cutoff policy that does not require tracking or communicating any ages. Instead, entangled links are discarded with a fixed probability after each round of HEG and swaps, regardless of their ages.
We note that the probabilistic cutoff policy abandons any attempt at strict control of the fidelity. Indeed, there is a nonzero probability that links in the repeater chain decohere completely, while high-fidelity links may be discarded. One may wonder what the cost is of giving up this control of the fidelity in terms of entanglement distribution performance.
We address this question by benchmarking the probabilistic cutoff policy against a deterministic cutoff policy. Our main findings are the following:
\begin{itemize}
    \item \textbf{Rate and fidelity.} For the same entanglement generation rate, the probabilistic cutoff policy generally yields a lower end-to-end fidelity than the deterministic cutoff policy. But when the average end-to-end fidelity is required to be above some high threshold value, then the probabilistic cutoff policy can satisfy this requirement with a higher rate than the deterministic cutoff policy. In \cref{section: rate and fidelity} we show an example of a three-node chain for which the rate can be improved by up to a factor of approximately $1.5$.
    \item \textbf{Secret-key rate.} In repeater chains with few nodes or high link generation probability, the probabilistic cutoff policy can provide secret-key rates of the same order of magnitude as the deterministic cutoff policy. In the parameter regime explored in \cref{section:SKR}, the secret-key rates for chains of three, four, and five nodes are at least $0.53$, $0.28$, and $0.15$ times that of the deterministic cutoff policy, respectively.
\end{itemize}

The rest of paper is structured as follows. In \cref{section: model} we provide the details of our repeater chain model, the definition of the probabilistic cutoff policy, and the definition of the deterministic cutoff policy used as a benchmark. In \cref{section: rate and fidelity} we compare the two policies in terms of end-to-end rates and fidelities. We make a comparison in terms of the secret-key rate in \cref{section:SKR}. We end the main text with a discussion in \cref{section: discussion}. The methods that we have used to compute repeater chain rates and fidelities are based on Markov chain models that are explained in detail in \cref{appendix: Markov chain models}.

\begin{figure}
    \centering
    \includegraphics[width=1\linewidth]{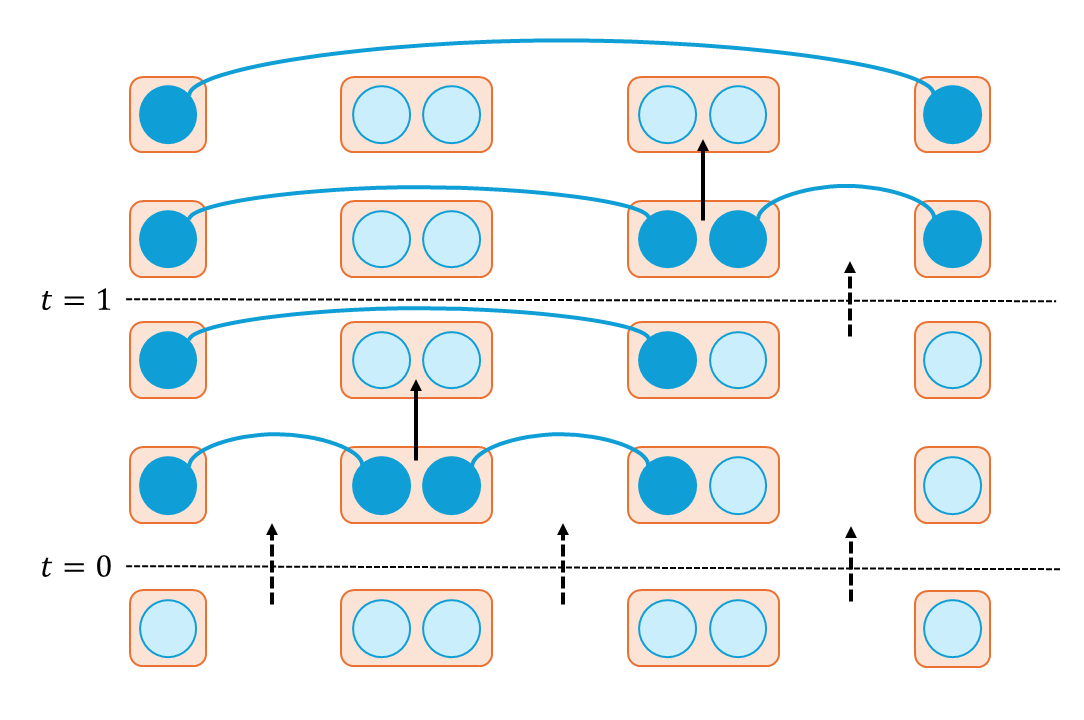}
    \caption{\textbf{Schematic representation of a quantum repeater chain.} A quantum repeater chain generates an entangled link between two end nodes by using heralded entanglement generation (HEG) and entanglement swaps. Dashed arrows indicate an HEG attempt and solid arrows indicate swaps. Qubits are depicted as blue circles. Qubits that share an entangled link are filled in and connected by an arc.}
    \label{fig: repeater chain schematic}
\end{figure}

\section{Model and cutoff policies}\label{section: model}
In this section, we provide the details of the quantum repeater chain model and define the cutoff policies. \cref{fig: repeater chain schematic} gives a schematic overview of the repeater chain protocol. The model parameters are summarized in \cref{table: parameters}. 

\textbf{The homogeneous quantum repeater chain.} We schematically visualize the repeater chain to be oriented horizontally as in \cref{fig: repeater chain schematic}. Every repeater node has one qubit that can be used to share entanglement with a node to its left and one qubit that can be used for a node to its right. Each of the end nodes has a single qubit that can be used to share entanglement with its neighbouring repeater node. We denote the total number of nodes in the chain, including both repeater nodes and end nodes, by $n_\mathrm{node}$. Other than the distinction between end nodes and repeater nodes, we assume the repeater chain to be homogeneous. This means that HEG between every pair of neighbouring nodes produces the same state with the same probability, every swap succeeds with the same probability and every qubit has the same coherence time. 

\textbf{Time step.} 
The repeater chain evolves in discrete time steps. Every time step consists of three phases. The first phase is for HEG attempts and the second phase for entanglement swaps. The third phase is for the application of the cutoff policy. 

\textbf{Heralded entanglement generation.}
The probability that HEG successfully produces a link between a pair of neighbouring nodes during the first phase of the time step is called the \emph{elementary link generation probability}. It is denoted $p_\mathrm{g}$. Even though the probability that a single HEG attempt succeeds may be on the order of $\sim 10^{-5}$, the probability $p_\mathrm{g}$ can be much higher by doing a batch of many such attempts during the HEG phase of the time step \cite{humphreys2018deterministic}. In the absence of noise, the link produced through HEG can be a maximally entangled state $\ket{\Phi^+}=\left(\ket{00}+\ket{11}\right)/\sqrt{2}$, when the Barrett-Kok protocol is used \cite{barrett2005efficient}, or very close to a maximally entangled state, when the protocol by Cirac et al. \cite{cabrillo1999creation} is used. 

\textbf{Werner states.} Various inefficiencies and noises lead to infidelities. We take this into account by assuming that a successful HEG attempt produces a Werner state \cite{werner1989quantum}.
A Werner state is a maximally entangled state that has been subjected to a depolarizing noise channel. On a two-qubit state $\rho$ a depolarizing noise channel $\mathcal{N}_\lambda$ with \emph{depolarizing parameter} $\lambda$ acts as
\begin{equation}\label{eq: depolarizing channel}
    \mathcal{N}_\lambda(\rho) = \lambda \rho + \frac{1-\lambda}{4}\II_4,
\end{equation}
where $\II_4$ is the four-dimensional identity matrix. 
Depolarizing noise is a worst-case noise model \cite{dur2005standard} in the sense that any noise channel can be made into a depolarizing noise channel by introducing additional noise in a process called twirling \cite{bennett1996mixed, horodecki1999general}.
We will consider Werner states of the form 
\begin{equation}\label{eq:werner state}
    \rho_{w} = w \ketbra{\Phi^+}{\Phi^+}+\frac{1-w}{4}\II_4,
\end{equation}
where $w\in [-1/3,1]$ is the referred to as the \emph{Werner parameter}. The fidelity $F(\rho_w)$ of the Werner state $\rho_w$ to the maximally entangled state $\ket{\Phi^+}$ is expressed in terms of the Werner parameter as $F(\rho_w)=\flatfrac{(1+3w)}{4}$. We denote the Werner parameter of a freshly generated link by $w_0$. 

\textbf{Memory decoherence.} Decoherence of a link stored in the quantum memories is modeled by a depolarizing noise channel with depolarizing parameter $\lambda = e^{-\flatfrac{1}{\tau_\mathrm{coh}}}$, where $\tau_\mathrm{coh}$ is the effective two-qubit coherence time. We express $\tau_\mathrm{coh}$ in units of repeater chain time steps. Note that the action of the depolarizing channel on a Werner state $\rho_w$ is simply to multiply the Werner parameter by the depolarizing parameter \cite{munro2015inside}.

\begin{figure}
    \centering
    \includegraphics[width=0.8\linewidth]{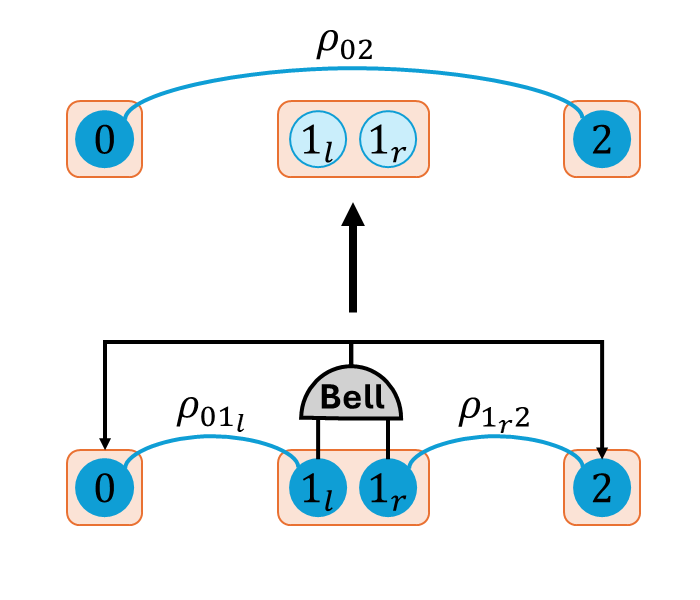}
    \caption{\textbf{An entanglement swap.} The qubits $1_l$ and $1_r$ are measured in the Bell basis. The measurement outcome must be communicated to nodes $0$ and $2$. By local operations at nodes $0$ and $2$ the entangled link $\rho_{02}$ is created.}
    \label{fig: entanglement swap}
\end{figure}

\textbf{Entanglement swaps.}
After the HEG phase follows the swap phase. An entanglement swap \cite{zukowski1993event, duan2001long} of a link $\rho_{01_l}$ between qubit $0$ on node $0$ and qubit $1_l$ on node $1$ and a link $\rho_{1_r2}$ between qubit $1_r$ on node $1$ and qubit $2$ on node $2$, as depicted in \cref{fig: entanglement swap}, proceeds as follows. A Bell state measurement is performed on the two qubits $1_l$ and $1_r$, and the outcome is communicated to nodes $0$ and $2$. Based on the outcome these nodes can then perform local operations that yield the entangled state $\rho_{02}$. If $\rho_{01_l}$ and $\rho_{1_r2}$ are Werner states with Werner parameters $w_{01_l}$ and $w_{1_r2}$, respectively, then the resulting state $\rho_{02}$ between nodes $0$ and $2$ is also a Werner state, and its Werner parameter is $w_{02}=w_{01_l}w_{1_r2}$ \cite{munro2015inside}. 

However, in some physical implementations of the Bell state measurement not all outcomes can be used to perform the required local operations for obtaining the entangled state $\rho_{02}$ \cite{munro2015inside}. Hence, an entanglement swap generally succeeds only probabilistically, with \emph{swap success probability} $p_\mathrm{s}\in (0,1]$. If the swap fails, then the links $\rho_{01_l}$ and $\rho_{1_r2}$ are lost because the measurement on qubits $1_l$ and $1_r$ has collapsed their state. 

\textbf{The age of a link.}
We can now properly introduce the notion of the age of a link. First, we say that elementary links are initially created with age $0$. Second, when the depolarizing channel $\mathcal{N}_{e^{-1/\tau_\mathrm{coh}}}$ is applied to a link, then its age is increased by $1$. An elementary link of age $t$ thus has Werner parameter $w_0 e^{-t/\tau_\mathrm{coh}}$. Finally, when links are swapped, their Werner parameters multiply so that their ages have to be added.

\textbf{Swap-asap policy.}
We assume that the swap-asap policy is used. In this policy, each node performs a swap as soon as both entangled links are available. The swap-asap policy minimizes the number of quantum memories that are needed to store entanglement, which helps to reduce noise due to memory decoherence \cite{kamin2023exact}. However, if swaps succeed with probability $p_\mathrm{s}<1$, then the success probability of simultaneous swaps between $k$ adjacent links is exponentially suppressed as $p_\mathrm{s}^k$. It has been shown in \cite{inesta2023optimal} that, for a particular deterministic cutoff policy, swap-asap performs well with respect to the optimal swapping policy when $p_\mathrm{s}$ is close to $1$. Here we consider $p_\mathrm{s}=1$.

\textbf{Probabilistic cutoff policy.}
If the end-to-end link has not been generated after the swaps, then in the third phase of the time step the remaining links in the chain are subjected to a cutoff policy. Previous work has used deterministic cutoff policies that discard links based on their age \cite{li2021efficient,haldar2023fast,kamin2023exact,inesta2023optimal}. We propose a probabilistic cutoff policy that does not require knowledge about the ages of links. The probabilistic cutoff policy is determined by a single, globally-defined, continuous parameter $p_\mathrm{c}\in [0,1]$, called the \emph{cutoff probability}. Every link left in the chain after the swap phase is discarded with probability $p_\mathrm{c}$.

\textbf{Deterministic cutoff policy.}
In the next sections we will benchmark the probabilistic cutoff policy against a deterministic cutoff policy. There are many possible choices of deterministic cutoff policies.  We consider a deterministic cutoff policy whose structure is most similar to the probabilistic cutoff policy. It is defined by a single, globally-defined, discrete parameter $t_\mathrm{c}\in \{0,1,\dots,\infty\}$, called the \emph{cutoff time}. The deterministic cutoff policy proceeds by discarding every link that remains after the swap phase whose age is greater than or equal to $t_\mathrm{c}$. Note that the probabilistic and deterministic policies defined by $p_\mathrm{c}=1$ and $t_\mathrm{c}=0$ both correspond to never storing any links, while $p_\mathrm{c}=0$ and $t_\mathrm{c}=\infty$ never discard a link. 

\textbf{Additional assumptions.} 
A final assumption that we make for chains of five or more nodes is that HEG attempts are not made between nodes ``under" an already existing link. For example, if there is a link between nodes $0$ and $3$, then we assume that there are no attempts to generate a link between nodes $1$ and $2$. This assumption only has to be made for chains of five or more nodes, because for chains of three or four nodes it is never possible to generate a link ``under" an already existing link. It ensures that after consumption of the end-to-end link the repeater chain returns to the empty state, i.e. the state in which there are no links in the chain. Consequently, all end-to-end links are independent, which will be important for the definitions of the rate and fidelity in the next section. This assumption also simplifies the analysis of the five-node chain. However, in chains with many more nodes this assumption may become unreasonable since the rate may benefit a lot from generating links ``under" already existing links. 



\begin{table*}[t]
\begin{center}
\renewcommand{\arraystretch}{1.5} 
\begin{tabular}{ | c | c | c | }
\hline
Parameter & Domain & Description \\
\hline \hline
$n_\mathrm{node}$ & $\{3,4,5,\dots\}$ & Number of nodes in the chain, including the two-end nodes \\
\hline
$p_\mathrm{g}$ & $(0,1]$ & Success probability of a round of HEG \\
\hline
$\tau_\mathrm{coh}$ & $[0,\infty)$ & Effective two-qubit coherence time of link stored in memory, in units of the time step \\
\hline
$p_\mathrm{c}$ & $[0,1]$ & Cutoff probability \\
\hline
$t_\mathrm{c}$ & $\{0,1,\dots,\infty\}$ & Cutoff time \\
\hline
$w_0$ & $(0,1]$ & Werner parameter of link resulting from HEG\\
\hline
$p_\mathrm{s}$ & $(0,1]$ & Swap success probability\\
\hline
\end{tabular}
\caption{Parameters of the discrete-time, homogeneous quantum repeater chain model.}
\label{table: parameters}
\end{center}
\end{table*}

\section{Rate and Fidelity}\label{section: rate and fidelity}
In this section we study the end-to-end rate and fidelity in quantum repeater chains with probabilistic or deterministic cutoff policy. The rate quantifies how fast end-to-end links are established, while the fidelity is a measure for their quality. The precise definitions within the model of \cref{section: model} are provided below.

Let us first define the \textit{delivery time} $T$ to be the number of time steps needed to generate the end-to-end link when starting from an empty repeater chain. We denote the expected delivery time as $\overline{T}:= \EE[T]$. The \textit{rate} $R$ is then defined to be the inverse of the expected delivery time. That is,
\begin{equation}\label{eq: definition rate}
    R := \overline{T}^{-1}.
\end{equation}
Since the repeater chain returns to the empty state after every end-to-end link it generates, $\overline{T}$ gives the expected time between the production of two end-to-end links. The rate thus equals the frequency at which end-to-end links are produced on average.

The age of the end-to-end link is also a random variable. In the probabilistic cutoff policy the age of the end-to-end link is not known because link ages are not tracked. The quantum state $\rho_\mathrm{e2e}$ that describes the end-to-end link is therefore a mixed state over all possible ages of the end-to-end link. If the end-to-end link has age $t$ with probability $p_t$, then 
\begin{equation}\label{eq: rho e2e}
    \rho_\mathrm{e2e}=\sum_{t\geq 0}p_t\rho_{w_t} = \rho_{\overline{w}},
\end{equation}
where $\rho_{w_t}$ is the Werner state corresponding to an end-to-end link of age $t$ and the second equality follows from the general form of the Werner state in \cref{eq:werner state}.
In other words, the end-to-end link is itself a Werner state, and its Werner parameter is the expected Werner parameter $\overline{w} := \sum_{t\geq 0} p_t w_t$ over all possible ages. The fidelity $F$ is defined to be the fidelity of $\rho_\mathrm{e2e}$ to the maximally entangled state $\ket{\Phi^+}$. In terms of the expected Werner parameter $\overline{w}$ it is given by 
\begin{equation}\label{eq: definition fidelity}
    F=\frac{1+3\overline{w}}{4}.
\end{equation}

In the deterministic cutoff policy the age of the end-to-end link is known. The end-to-end link is then one of the Werner states $\rho_{w_t}$, instead of the mixture $\rho_{\mathrm{e2e}}$ over all the $\rho_{w_t}$. However, the fidelity of the mixed state $\rho_\mathrm{e2e}$ can be interpreted as the \textit{average} fidelity of the end-to-end link and it will be used as the single fidelity metric to describe the end-to-end link in the deterministic cutoff policy. This allows for a direct comparison to the end-to-end link in the probabilistic cutoff policy.

Using Markov chain theory \cite[Theorems 1.42 and 1.43]{serfozo2009basics}, the expected delivery time and expected Werner parameter can be computed exactly in chains of three, four and five nodes. We show in \cref{appendix: Markov chain models} that for the deterministic cutoff policy the rate and expected Werner parameter in a three-node chain can be expressed analytically as
\begin{equation}\label{eq: analytic rate}
    R = p_\mathrm{g}^2p_\mathrm{s} \frac{1+2A_1(t_\mathrm{c})}{1+2p_\mathrm{g}A_1(t_\mathrm{c})},
\end{equation}
and
\begin{equation}\label{eq: analytic werner}
    \overline{w} = w_0^2 \frac{1+2A_\lambda(t_\mathrm{c})}{1+2A_1(t_\mathrm{c})},
\end{equation}
where 
\begin{equation}
    A_\lambda(t_\mathrm{c}) = [1-(1-p_\mathrm{g})^{t_\mathrm{c}}\lambda^{t_\mathrm{c}}]\frac{(1-p_\mathrm{g})\lambda }{1-(1-p_\mathrm{g})\lambda}.
\end{equation}
The subscript $\lambda$ refers to the depolarizing parameter $\lambda = e^{-1/\tau_\mathrm{coh}}$ of the quantum memories. Moreover, we show that replacing $A_\lambda(t_\mathrm{c})$ by 
\begin{equation}
    B_\lambda(p_\mathrm{c}) =  \frac{(1-p_\mathrm{g})(1-p_\mathrm{c})\lambda}{1-(1-p_\mathrm{g})(1-p_\mathrm{c})\lambda}
\end{equation}
in \cref{eq: analytic rate} and \cref{eq: analytic werner} gives the rate and expected Werner parameter for the probabilistic cutoff policy in a three-node chain.

For four- and five-node chains we have not found closed-form analytic expressions. Instead, we have linear systems that can be solved exactly on a computer, and from whose solutions the expected delivery time and expected Werner parameter can be derived. For the deterministic cutoff policy we have adapted the Markov chain methods of \cite{shchukin2019waiting}. For the probabilistic cutoff policy we have developed a non-trivial extension of these Markov chain methods. This is presented in detail in \cref{appendix: Markov chain models}. 

We compare the rates and fidelities that are achieved for different values of $p_\mathrm{c}$ and $t_\mathrm{c}$ in \cref{fig: rate fidelity curve}. The rate-fidelity curves displayed there correspond to a particular choice of parameters, but different choices lead to similar results (see \cref{appendix: RF supplement}). In particular, it can be seen that the probabilistic cutoff policy yields lower fidelity than the deterministic cutoff policy when compared at the same rate. This result is consistent with the intuition that in the probabilistic cutoff policy we have given up strict control of the fidelity: Links are allowed to become arbitrarily old, while in the deterministic cutoff policy there is a finite maximum age. The ensemble of end-to-end links that constitute $\rho_\mathrm{e2e}$ for some particular value of the rate should thus have a higher expected age in the case of the probabilistic cutoff policy. A proof of this fact for three-node chains is given in \cref{appendix: RF supplement} from the analytic expressions of $R$ and $\overline{w}$.

\begin{figure}
    \centering
    \includegraphics[width=0.9\linewidth]{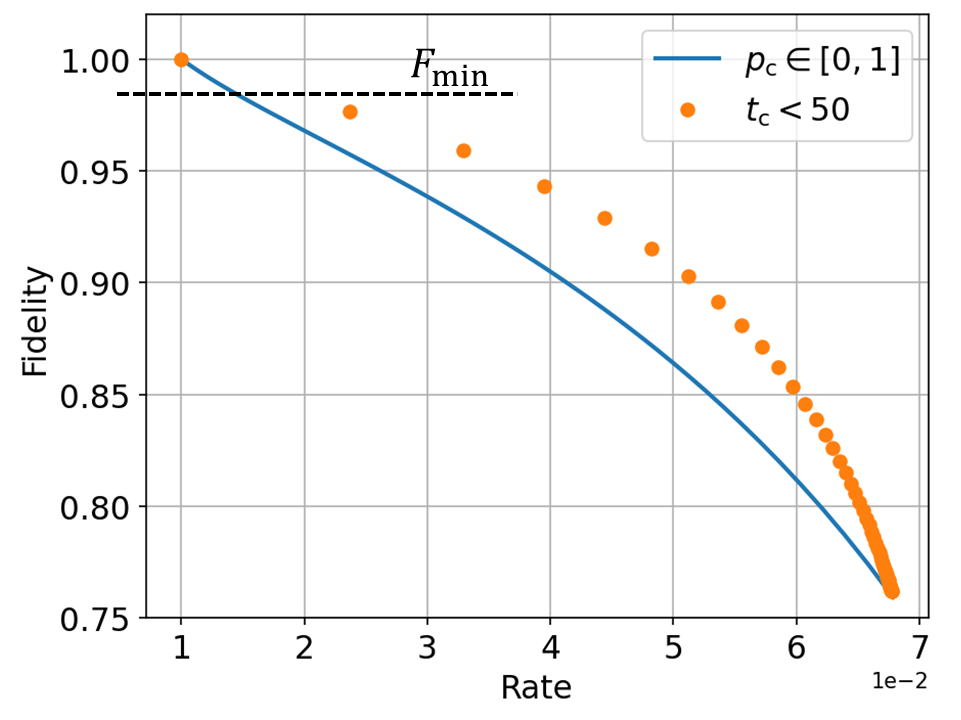}
    \caption{\textbf{Deterministic cutoffs provide higher fidelities than probabilistic cutoffs.} Rates and fidelities of quantum repeater chain with probabilistic or deterministic cutoff policy. Results shown are for parameter values $n_\mathrm{node}=3$, $p_\mathrm{g}=0.1$,  $\tau_\mathrm{coh}=20$, $w_0=1$ and $p_\mathrm{s}=1$. Horizontal dashed line indicates comparison of rates above minimum threshold fidelity $F_{\mathrm{min}}$.}
    \label{fig: rate fidelity curve}
\end{figure}

Applications of a quantum network may require the end-to-end link to have fidelity above some minimum threshold $F_\mathrm{min}$ \cite{dahlberg2019link,kozlowski2020designing}. In some special cases, the probabilistic cutoff policy can satisfy such a requirement at a higher rate than the deterministic cutoff policy. In particular, it does so whenever the minimum threshold fidelity is so high that the deterministic cutoff policy can only satisfy it with $t_\mathrm{c}=0$. We have indicated the situation in \cref{fig: rate fidelity curve}. The probabilistic cutoff policy provides at least the same rate as the deterministic cutoff policy by setting $p_\mathrm{c}=1$. Subsequently, the cutoff probability can be continuously decreased until the fidelity exactly coincides with the threshold. In the example of \cref{fig: rate fidelity curve}, this improves the rate of the probabilistic cutoff policy over what is possible with the deterministic cutoff policy by a factor of approximately $1.5$. 

The above advantage of the probabilistic cutoff policy crucially relies on the fact that the cutoff probability is continuous, while the cutoff time is discrete. Moreover, the continuity is only sure to yield an advantage in the case that $t_\mathrm{c}=0$ is the only feasible cutoff time. When a positive cutoff time $t_\mathrm{c}>0$ is feasible the rate of $t_\mathrm{c}$ can only be matched by the probabilistic cutoff policy at a lower fidelity, which in general may not be above the threshold.  

\section{Secret-Key Rate}\label{section:SKR}
The secret-key rate (SKR) combines the rate and fidelity into a single number with an operational interpretation. Namely, the number of bits of secret key per unit of time that two parties produce in a quantum key distribution protocol, such as entanglement based BB84 \cite{bennett2014quantum, ekert1991quantum,bennett1992quantum}. In this section we compare the performance of the probabilistic and deterministic cutoff policies in terms of the secret-key rate.

For a quantum repeater chain that produces end-to-end links with expected Werner parameter $\overline{w}$ at rate $R$ to run a version of the BB84 quantum key distribution protocol, the secret-key rate is given by
\begin{equation}\label{eq:skr definition}
    \SKR = R\cdot \SKF(\overline{w}),
\end{equation}
where
\begin{equation}\label{eq:skf definition}
    \SKF(\overline{w}) = \max\left\{1-2\,h\left(\frac{1-\overline{w}}{2}\right) ,0 \right\}
\end{equation}
is the so-called secret-key fraction and $h(x)=-x\log_2 x -(1-x)\log_2(1-x)$ is the binary entropy \cite{renner2008security, gottesman2004security}. The secret-key fraction is the number of bits of secret-key that can be extracted per end-to-end link. It is zero as long as $\overline{w}\lesssim 0.78$, and then increases monotonically to become one at $\overline{w}=1$. We consider $w_0=1$ throughout this section. For fixed hardware parameters $n_\mathrm{node}$, $p_\mathrm{g}$, $p_\mathrm{s}$ and $\tau_\mathrm{coh}$, the secret-key rate can be maximized over the cutoff probability $p_\mathrm{c}$ of the probabilistic cutoff policy or cutoff time $t_\mathrm{c}$ of the deterministic cutoff policy. We denote the results by $\max_{p_\mathrm{c}} \SKR$ and $\max_{t_\mathrm{c}} \SKR$, respectively. 

The maximized secret-key rates for the two cutoff policies in chains of three to five nodes with deterministic swaps are compared in \cref{fig: skr pg} as a function of the elementary link generation probability. It can be seen that the difference between the probabilistic cutoff policy and deterministic cutoff policy increases as $p_\mathrm{g}$ decreases. At the smallest value considered, $p_\mathrm{g}=10^{-3}$, the ratios $\flatfrac{\max_{p_\mathrm{c}}\SKR}{\max_{t_\mathrm{c}}\SKR}$ are largest. They are $0.53$ for $n_\mathrm{node}=3$, $0.28$ for $n_\mathrm{node}=4$ and $0.15$ for $n_\mathrm{node}=5$. As we discuss in \cref{appendix: SKR}, the coherence time has been chosen to make these ratios as small as possible (see \cref{fig: skr tau_coh ratios}), in the parameter regime explored. We would not expect them to dramatically decrease when $p_\mathrm{g}$ is decreased further, since they already exhibit convergence around $p_\mathrm{g}=10^{-3}$ (see \cref{fig: skr pg ratios}). Hence, these ratios give an indication of the worst-case secret-key rate loss of the probabilistic cutoff policy with respect to the deterministic cutoff policy. 

Even though the absolute values of the secret-key are sensitive to the swap success probability, the ratio of the maximized secret-key rates for $n_\mathrm{node}=3$ is completely independent of the swap success probability. This is because $p_\mathrm{s}$ only appears as an overall scaling of the rate in \cref{eq: analytic rate}. For $n_\mathrm{node}=4$ and $n_\mathrm{node}=5$ we have found numerically that the ratio of maximized secret-key rates is still largely insensitive to the swap success probability. For example, probabilistic swaps with $p_\mathrm{s}=0.5$ yield the same ratios of the maximized secret-key rates for $p_\mathrm{g}=10^{-3}$ and $\tau_\mathrm{coh}=50$ as listed above for deterministic swaps.

\begin{figure}
    \centering
    \includegraphics[width=0.9\linewidth]{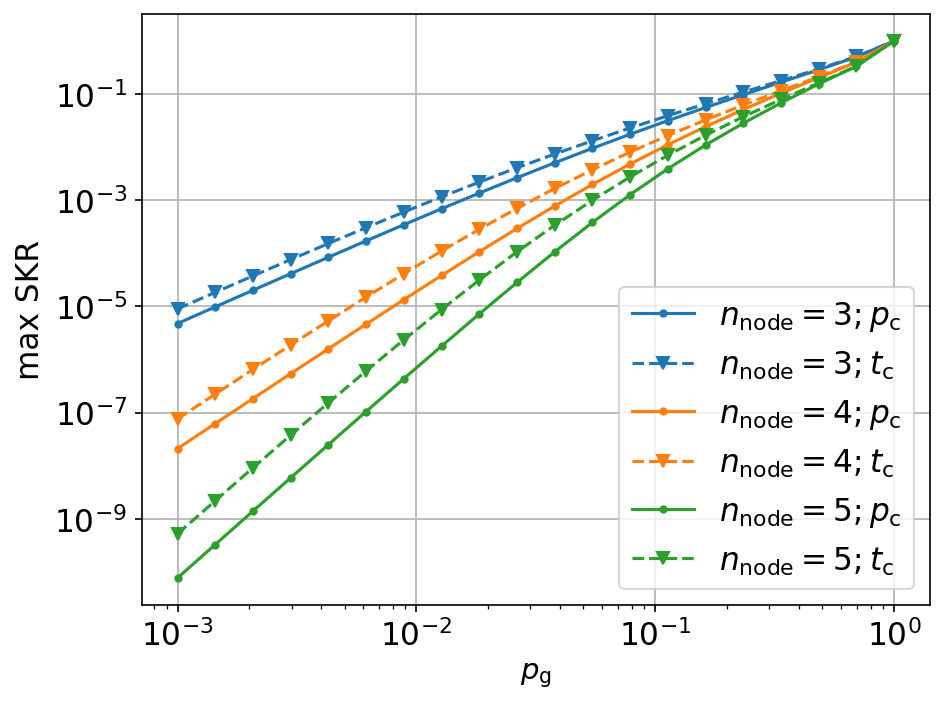}
    \caption{\textbf{Probabilistic cutoffs provide secret-key rates of the same order of magnitude as deterministic cutoffs over a large range of link generation probabilities when $n_\mathrm{node}\leq 5$. }Maximized secret-key rates $\max_{p_\mathrm{c}}\SKR$ and $\max_{t_\mathrm{c}}\SKR$ for probabilistic ($p_\mathrm{c}$) and deterministic ($t_\mathrm{c}$) cutoff policy, respectively, are shown as function of elementary link generation probability $p_\mathrm{g}\in [10^{-3},1]$ for $\tau_\mathrm{coh}=50$ and $p_\mathrm{s}=1$.}
    \label{fig: skr pg}
\end{figure}

Consider now the regime of large elementary link generation probability. As can be seen in \cref{fig: skr pg} (or \cref{fig: skr pg ratios}) the ratio between between the maximized secret-key rates of the two cutoff policies in that regime is much closer to $1$ than in the worst case. We use a Monte Carlo simulation to show how the secret-key rates compare in this regime when the number of nodes is increased beyond five. We restrict to the case of deterministic swaps since our Monte Carlo simulation method becomes inefficient for probabilistic swaps, and we have also seen above that the comparison of the two cutoff policies is largely independent of the swap success probability, at least up to five node chains.

We show the maximized secret-key rates from the probabilistic cutoff policy and deterministic cutoff policy in chains up to ten nodes for $p_\mathrm{g}=0.25$ and $\tau_\mathrm{coh}=50$ in \cref{fig: skr n}. We also show the secret-key rate that would have been achieved with a trivial cutoff policy. The two trivial policies are to either never discard a link ($p_\mathrm{c}=0$) or to never store a link ($p_\mathrm{c}=1$). It turns out that, in the case shown, the better trivial cutoff policy for $n_\mathrm{node}\leq 6$ is to never discard any links. For $n_\mathrm{node}>6$, however, that policy would yield zero secret-key rate because the expected Werner parameter drops below $0.78$. The only possible trivial cutoff policy for $n_\mathrm{node}>6$ is therefore to discard every link. It can be seen in \cref{fig: skr n} that this causes the secret-key rate of the trivial cutoff policy to become orders of magnitude lower than that of the deterministic cutoff policy. The probabilistic cutoff policy, on the other hand, continues to yield secret-key rates of the same order of magnitude as the deterministic cutoff policy, even for $n_\mathrm{node}>6$. 

For a larger elementary link generation probability $p_\mathrm{g}=0.5$ and the same coherence time $\tau_\mathrm{coh}=50$ we have found that the optimal cutoff probability is $p_\mathrm{c}=0$ for all $n_\mathrm{node}\leq 10$. In that case the optimal probabilistic cutoff policy is thus trivial. For a lower coherence time of $\tau_\mathrm{coh}=20$ the optimal probabilistic cutoff policy for $p_\mathrm{g}=0.5$ is non-trivial again. It provides similar secret-key rates to the deterministic cutoff policy for all $n_\mathrm{node}\leq 10$, but the improvement in secret-key rate over the best trivial policy is more modest than what was observed in \cref{fig: skr n}. 

\begin{figure}
    \centering
    \includegraphics[width=0.9\linewidth]{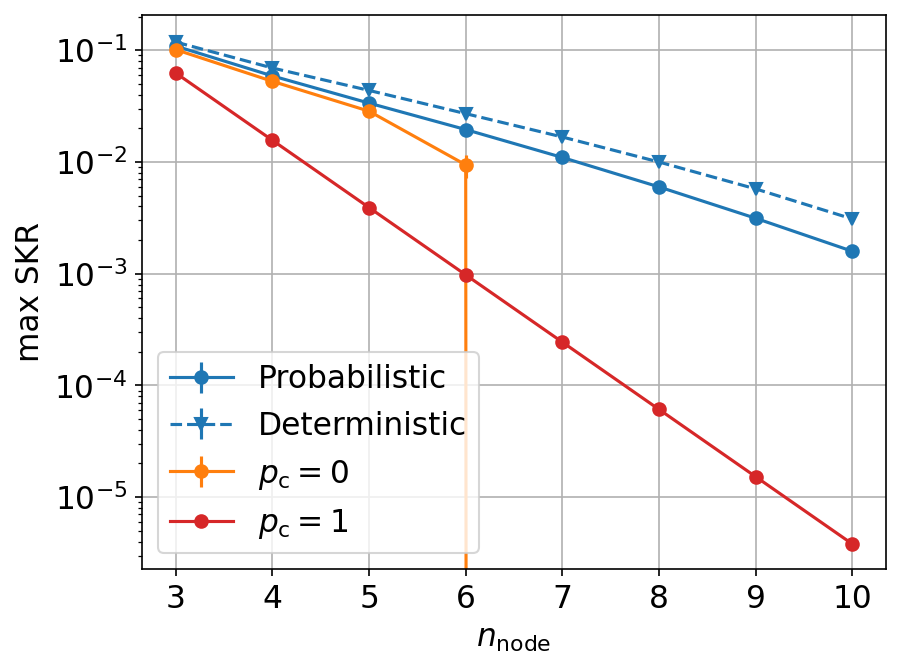}
    \caption{\textbf{Probabilistic cutoffs provide orders of magnitude higher secret-key rate than trivial cutoff policies, and are competitive with deterministic cutoffs when link generation probability is large.} Maximized secret-key rates $\max_{p_\mathrm{c}}\SKR$ and $\max_{t_\mathrm{c}}\SKR$ for probabilistic and deterministic cutoff policy, respectively, are shown as function of $n_\mathrm{node}=3,4,\dots,10$ for chains with $p_\mathrm{g}=0.25$, $\tau_\mathrm{coh}=50$ and $p_\mathrm{s}=1$. The secret-key rates of the trivial policies to never discard a link ($p_\mathrm{c}=0$) or never store a link $(p_\mathrm{c}=1$) are also shown. No secret-key can be generated with $p_\mathrm{c}=0$ for $n_\mathrm{node}>6$. Errorbars show one standard deviation, but almost all are smaller than marker (see \cref{appendix: SKR} for details).}
    \label{fig: skr n}
\end{figure}

\section{Discussion}\label{section: discussion}
In this paper we have introduced the probabilistic cutoff policy for quantum repeater chains because it does not require tracking link ages, contrary to previously studied deterministic cutoff policies. We have shown that even though the probabilistic cutoff policy keeps track of less state, it can still provide comparable performance to a deterministic cutoff policy, especially when chains are short or have high link generation probability.

In our present model every node has only one qubit to share entanglement with each of its neighbours, so that there can be at most one link between a pair of nodes. On the other hand, in a multiplexed repeater chain the nodes will support many qubits and more than one link can be present between a pair of nodes \cite{lee2022quantum}. It could thus be especially advantageous if the ages of all these links do not have to be tracked. Moreover, the probability that at least one link is generated between neighbouring nodes in a multiplexed chain is expected to be high. According to our results, this is the regime in which the probabilistic cutoff policy is most competitive with the deterministic cutoff policy.

Finally, note that the deterministic cutoff policy which we have used as a benchmark has been chosen for its similarity to the probabilistic cutoff policy, but other deterministic cutoff policies are possible. In \cref{appendix: other cutoff policies} we discuss a variation on the deterministic cutoff policy which consists in adding a post-selection: end-to-end links with age greater than $t_\mathrm{c}$ are discarded. We show in the appendix that this post-selection can lead to unexpectedly low rates for the deterministic cutoff policy in certain parameter regimes. Adding global coordination between the nodes could prevent these low rates, but may come at a higher cost in classical communication. However, even for the probabilistic and deterministic cutoff policies considered in this work it is not obvious what the effects of introducing classical communication into the model will be. Both policies require classical communication to herald entanglement generation and entanglement swaps. The deterministic cutoff policy requires classical communication to update link ages. The probabilistic cutoff policy does not require this classical communication of link ages, but it does require heralding cutoffs. We leave the extension of the model to account for classical communication to future work.

\section*{Code and data availability}
Our code and data can be found in the following GitLab repository: \url{https://gitlab.tudelft.nl/wehner-research/probabilistic_cutoffs_in_homogeneous_quantum_repeater_chains}

\printbibliography

\section*{Acknowledgements}
J.G. and S.W. acknowledge funding from an NWO Vici grant. A.G.I. acknowledges financial support from the Netherlands Organisation for Scientific Research (NWO/OCW), as part of the Frontiers of Nanoscience program. We thank Sounak Kar and Francisco Ferreira Da Silva for discussions and feedback.

\section*{Author Contributions}
S.W. and A.G.I. proposed to consider probabilistic cutoffs for quantum repeater chains. J.G. analyzed the system and was the main writer of this paper. The Monte Carlo simulation has been written in collaboration with S.H.. A.G.I. and S.W. provided active feedback at every stage of the project.

\begin{appendices}
\counterwithin{figure}{section}
\section{Markov chain models}\label{appendix: Markov chain models}
In this appendix we present the exact methods used for the computation of the rate and fidelity in repeater chains of three, four, and five nodes with probabilistic or deterministic cutoffs, as defined in \cref{section: model}. We show how to formulate the discrete-time repeater chain model as a discrete-time Markov chain model. The states of the Markov chain correspond to the different configurations of links in the repeater chain. Using Markov chain theory the computation of the expected delivery time and expected Werner parameter, and hence of the rate and fidelity, is reduced to solving linear systems of equations. In the three-node case these systems can be solved analytically, while for four and five nodes we use numerical methods. The Markov chain model we use for deterministic cutoffs has been adapted from \cite{shchukin2019waiting}, while a new model is developed to handle the case of probabilistic cutoffs. 
All computations have been implemented in Python and are available at \cite{Grimbergen2026}.

\subsection{Deterministic cutoffs}
In this section we present a discrete-time Markov chain model for quantum repeater chains with deterministic cutoffs. The use of such a Markov chain model for the computation of the expected delivery time was originally proposed in \cite{shchukin2019waiting}. We adapt the model of \cite{shchukin2019waiting} so that the expected Werner parameter can be computed as well. 

Recall that a Markov chain $(\mathcal{S}, P)$ consists of a state space $\mathcal{S}$ and a transition probability matrix $P$. The states of the Markov chain that models the quantum repeater chain describe the entangled links in the chain and their ages. The transition probability matrix gives the probabilities of states to evolve to other states during one repeater chain time step of HEG attempts, swaps and cutoffs. 

If the repeater chain is pictured as a linear graph so that each pair of neighbouring nodes is connected by an edge, then for a repeater chain with $n_\mathrm{node}$ nodes there are $n=n_\mathrm{node}-1$ edges. We refer to these edges as the \emph{segments} of the repeater chain. In a quantum repeater chain with $n=n_\mathrm{node}-1$ segments we label the states by tuples $\tvec=(t_0,\dots,t_{n-1})$. We use $t_i\geq 0$ to indicate that there is a link on segment $i$ with age $t_i$. If there is no link on segment $i$ we set $t_i=-1$. If adjacent links have been swapped to form an entangled link that spans multiple segments (i.e. that connects non-neighbouring nodes), then we record the age of this entangled link on its left most segment. The ages of all other segments of the entangled link are set to $0$. For example, if link $i$ and $i+1$ have been swapped to form an entangled link and the chain is otherwise empty, then we denote this by $$(-1,\dots,-1,t_i,0,-1,\dots,-1),$$ where $t_i$ is the age of the entangled link spanning segments $i$ and $i+1$. 

The element $P_{\tvec,\tvec'}$ of the transition probability matrix is equal to the probability that the repeater chain transitions to state $\tvec'$ in a single time step starting from state $\tvec$. The nonzero elements of $P$ can be found by considering the possible evolutions of a state through a single repeater chain time step. First, during the HEG phase, there is an HEG attempt for every $t_i=-1$ in $\tvec=(t_0,\dots,t_{n-1})$ that succeeds with probability $p_\mathrm{g}$. Then, for each possible outcome of the HEG attempts swaps are done according to swap-asap during the swap phase. Finally, for each outcome of the swap phase the cutoffs must be applied to find the state $\tvec'$ that is reached at the end of the time step. As an example, we provide the details of this procedure for the three-node chain at the end of this section.

We partition the state space $\mathcal{S}$ into transient and absorbing states,
\begin{equation}
    \mathcal{S} = \mathcal{S}_{\mathrm{transient}}\cup \mathcal{S}_{\mathrm{absorbing}}.
\end{equation}
The transient states are all states in which there is no end-to-end link and the absorbing states are the states in which there is an end-to-end link. Since we assume swap-asap the transient states are labeled by tuples $\tvec=(t_0,\dots, t_{n-1})$ in which $t_i=-1$ for at least one value of $i$, while the absorbing states are labeled by tuples $\tvec=(t,0,\dots,0)$, where $t$ is the age of the end-to-end link. 

\paragraph{Expected delivery time}
The expected delivery time $\overline{T}$ of producing the end-to-end link when starting from the empty repeater chain can be computed from the Markov chain model as the expected hitting time of $\mathcal{S}_\mathrm{absorbing}$ starting from the empty repeater chain state $\boldsymbol{-1}:=(-1,\dots,-1)$. In general, if $v_{\tvec}$ is the expected hitting time of $\mathcal{S}_\mathrm{absorbing}$ from any $\tvec \in \mathcal{S}_\mathrm{transient}$, then \cite[Theorem 1.43]{serfozo2009basics}
\begin{align}
    v_{\tvec} =1+\sum_{\tvec' \in \mathcal{S}_\mathrm{transient}} P_{\tvec,\tvec'}v_{\tvec'}. \label{eq: linear system delivery time deterministic}
\end{align}
The linear system in \cref{eq: linear system delivery time deterministic} can be solved for all $v_{\tvec}$ with $\tvec\in \mathcal{S}_\mathrm{transient}$. The expected delivery time $\overline{T}$ of an end-to-end link starting from the empty chain is $v_{\boldsymbol{-1}}$.

\paragraph{Expected Werner parameter}
In the computation of the expected delivery time we compute the expected time to hit \emph{any} end-to-end state in $\mathcal{S}_\mathrm{absorbing}$. To compute the expected Werner parameter we must compute the probability to end up in \emph{a particular} end-to-end state in $\mathcal{S}_\mathrm{absorbing}$, because the Werner parameter of each end-to-end state is different. Let $\gamma_{\tvec,\hat\tvec}$ denote the probability that the Markov chain eventually ends up in the absorbing state $\hat\tvec\in \mathcal{S}_\mathrm{absorbing}$ when starting from $\tvec\in \mathcal{S}_\mathrm{transient}$, then \cite[Theorem 1.42]{serfozo2009basics}
\begin{align}
    \gamma_{\tvec,\hat\tvec} &= P_{\tvec,\hat\tvec} + \sum_{\tvec'\in \mathcal{S}_\mathrm{transient}} P_{\tvec,\tvec'}\gamma_{\tvec',\hat\tvec}.\label{eq: linear system werner parameter deterministic}
\end{align}
For each absorbing state $\hat\tvec\in \mathcal{S}_\mathrm{absorbing}$ the linear system in \cref{eq: linear system werner parameter deterministic} can be solved for all $\gamma_{\tvec,\hat\tvec}$ with $\tvec\in \mathcal{S}_\mathrm{transient}$ and the probability to end up in $\hat \tvec$ starting from the empty chain is recovered as $\gamma_{\boldsymbol{-1},\hat\tvec}$. The Werner parameter of the end-to-end link in state $\hat \tvec=(t,0,\dots,0)$ is given by $w(\hat\tvec) = w_0^{n}\lambda^t$, where $\lambda = e^{-1/\tau_\mathrm{coh}}$ is the depolarizing parameter, $w_0$ the Werner parameter of a freshly generated link, and $n=n_\mathrm{node}-1$ is the number of segments in the chain. The expected Werner parameter starting from the empty chain can then be computed as
\begin{align}
    \overline{w} &= \sum_{\hat \tvec \in \mathcal{S}_\mathrm{absorbing}} \gamma_{\boldsymbol{-1},\hat\tvec}w(\hat\tvec).\label{eq: werner deterministic}
\end{align}

The expected delivery time and Werner parameter can thus be computed by solving the linear systems \cref{eq: linear system delivery time deterministic} and \cref{eq: linear system werner parameter deterministic}, respectively. For chains of four and five nodes we have implemented an algorithm that builds the state space and the transition probability matrix, and then solves for the expected delivery time and expected Werner parameter in \cite{Grimbergen2026}. The linear systems are solved using the Python function \texttt{numpy.linalg.solve}. However, the size of these linear systems scales as $t_\mathrm{c}^2$, where $t_\mathrm{c}$ is the cutoff time. This is because for four and five node chains there can be two disjoint links in the chain whose ages range from $0$ to $t_\mathrm{c}$. For large values of the cutoff time computations thus become infeasible. Throughout this work we restrict to $t_\mathrm{c}<30$ in the case of four nodes, and to $t_\mathrm{c}<20$ in the case of five nodes. Part of the reason we have not considered more than five nodes is that the range of cutoff times would then have to be restricted even further. The other reason is that our methods for probabilistic cutoffs do not straightforwardly generalize to more than five nodes, as we discuss in the next section.

\paragraph{Three-node chains} For three-node chains the Markov chain model for quantum repeater chains with deterministic cutoffs is sufficiently simple so that we can analytically solve for the expected delivery time and expected Werner parameter. We will first simplify the notation for the state space by exploiting the symmetry of the chain. According to the general recipe above we should build the state space out of tuples $\tvec=(t_0,t_1)$. For cutoff time $t_\mathrm{c}$ the transient state space would be 
\begin{multline}
       \mathcal{S}_\mathrm{transient} = \{(-1,-1), (0,-1), (-1,0),\\ (1,-1), (-1,1),\dots,(t_\mathrm{c}-1,-1),(-1,t_\mathrm{c}-1)\}
\end{multline}
and the absorbing state space would be
\begin{equation}
    \mathcal{S}_\mathrm{absorbing} = \{(0,0), (1,0),\dots,(t_\mathrm{c},0)\}.
\end{equation}
Since the chain is homogeneous the states $(t,-1)$ and $(-1,t)$ are physically equivalent. They both describe a repeater chain with a link of age $t$ on one of the two segments. We can thus define the equivalent state spaces 
\begin{align}
    \mathcal{S}_\mathrm{transient} &= \{-1,0,1,\dots,t_\mathrm{c}-1\} \\
    \mathcal{S}_\mathrm{absorbing} &= \{\hat 0,\hat 1,\dots, \hat t_\mathrm{c}\},
\end{align}
where we have made the identifications $-1\equiv (-1,-1)$, $t\equiv (t,-1)\equiv (-1,t)$, and $\hat t \equiv (t,0)$. To be clear, $-1$ denotes the empty state, $t\geq 0$ denotes a link of age $t$ on either of the two segments, and the notations $\hat t$ with $t\geq 0$ denotes an end-to-end link with age $t$.

To compute the expected delivery time and expected Werner parameter let us consider the cases $t_\mathrm{c}=0$ and $t_\mathrm{c}>0$ separately.
For $t_\mathrm{c}=0$ the transient state space becomes $\mathcal{S}_\mathrm{transient}=\{-1\}$ and the absorbing state space is $\mathcal{S}_\mathrm{absorbing}=\{\hat 0\}$. The probability to transition from state $-1$ to state $\hat 0$ is
\begin{equation}
    P_{-1,\hat 0} = p_\mathrm{g}^2 p_\mathrm{s},
\end{equation}
because the end-to-end link is produced if and only if a link is generated in both segments, which has probability $p_\mathrm{g}^2$, and the subsequent swap is successful, which has probability $p_\mathrm{s}$. The probability to return to the empty state is
\begin{equation}
    P_{-1,-1}=1-P_{-1,\hat 0} = 1-p_\mathrm{g}^2p_\mathrm{s}.
\end{equation}
There are no other transition probabilities.
The linear system \cref{eq: linear system delivery time deterministic} for the hitting time $v_{-1}$ becomes
\begin{equation}
    v_{-1} = 1 + P_{-1,-1} v_{-1},
\end{equation}
which is solved for $v_{-1} = (1-P_{-1,-1})^{-1}=P_{-1,\hat 0}^{-1}=(p_\mathrm{g}^2p_\mathrm{s})^{-1}$. Hence, for $t_\mathrm{c}=0$ the expected delivery time in a three-node chain is $\overline{T} = (p_\mathrm{g}^2p_\mathrm{s})^{-1}$. Since $\{\hat 0\}$ is the only absorbing state the expected Werner parameter of the end-to-end link is simply the Werner parameter of state $\{\hat{0}\}$, which is $w(\hat 0) = w_0^2$.

Let us now turn to the case of $t_\mathrm{c}>0$. As we explain below, the nonzero transition probabilities out of state $-1$ are given by
\begin{align*}
    P_{-1,\hat{0}} &= p_\mathrm{g}^2 p_\mathrm{s} \\
    P_{-1,0} &= 2p_\mathrm{g}(1-p_\mathrm{g}) \\
    P_{-1,-1} &=  (1-p_\mathrm{g})^2 + p_\mathrm{g}^2(1-p_\mathrm{s}).
\end{align*}
Again, $P_{-1,\hat{0}}$ is the probability to generate an end-to-end link immediately from the empty repeater chain state. The probability $P_{-1,0}$ is the probability to generate a single link, either on the left or the right segment of the chain. Finally, $P_{-1,-1}$ is the probability to return to the empty chain, which can happen either when no links are generated, or when both links are generated and the swap between them fails.

If the chain is in state $t$ for $t\in \{0,\dots, t_\mathrm{c}-2\}$, then there is one link present of age $t$ and there are three possible outcomes of the time step. First, if the second link is not generated, then the link already present ages by one. Second, if the second link is generated and the swap succeeds, then an end-to-end link of age $t+1$ is produced. Finally, if the second link is generated but the swap fails, then the chain returns to the empty state. This leads to the transition probabilities 
\begin{align*}
    P_{t,t+1} &= 1-p_\mathrm{g} \\
    P_{t,\widehat{t+1}} &= p_\mathrm{g}p_\mathrm{s} \\
    P_{t,-1} &= p_\mathrm{g}(1-p_\mathrm{s})
\end{align*}
for $t\in \{0,\dots,t_\mathrm{c}-2\}$.
Finally, the nonzero transition probabilities from state $t_\mathrm{c}-1$ are given by
\begin{align*}
    P_{t_\mathrm{c}-1,\hat t_\mathrm{c}} &= p_\mathrm{g}p_\mathrm{s} \\
    P_{t_\mathrm{c}-1,-1} &= (1-p_\mathrm{g}) + p_\mathrm{g}(1-p_\mathrm{s}).
\end{align*}
The difference with the case $t\in\{0,\dots,t_\mathrm{c}-2\}$ is that for $t=t_\mathrm{c}-1$, the active link must be discarded when the second link is not generated. The Markov chain for $t_\mathrm{c}>0$ is depicted in \cref{fig: Markov chain}.

The expected delivery time of the end-to-end link can now be computed by solving the system \cref{eq: linear system delivery time deterministic} for $v_{-1}$. From \cref{eq: linear system delivery time deterministic} we have 
\begin{equation*}
    v_{t} =
    \begin{cases}
    1 + P_{-1,-1}v_{-1} + P_{-1,0}v_0  &\mathrm{if\,}t=-1 \\
    1+ P_{t,-1} v_{-1} + P_{t,t+1}v_{t+1}&\mathrm{if\,}t\in \{0,\dots, t_\mathrm{c}-2\} \\
    1 + P_{t_\mathrm{c}-1,-1}v_{-1} &\mathrm{if\,}t=t_\mathrm{c}-1
    \end{cases}.
\end{equation*}
An equation for $v_{-1}$ can be obtained by starting from the expression for $v_{-1}$ and recursively substituting $v_{t+1}$ into $v_t$. The result is
\begin{align*}
    v_{-1} &= 1+P_{-1, -1}v_{-1}+P_{{-1},0} v_{0}\\
    &= 1+P_{-1,0}+(P_{-1,-1}+P_{-1,0}P_{0,-1})v_{-1}\\
    &\qquad\qquad+P_{{-1},0}P_{0,1} v_{1}
    \\
    &= \dots \\
    &= 1+\sum_{t=0}^{t_\mathrm{c}-1}\prod_{t'=0}^tP_{t'-1,t'} \\
    &\qquad\qquad+\left( P_{-1,-1}+\sum_{t=0}^{t_\mathrm{c}-1}P_{t,-1}\prod_{t'=0}^tP_{t'-1,t'}\right)v_{-1}.
\end{align*}
The final equation can be solved for $v_{-1}$. After substituting the transition probabilities, and recalling that $\overline{T}=v_{-1}$, we find that the expected delivery time is given by
\begin{align}\label{eq: delivery time three node deterministic}
        \overline{T} 
        &= \frac{1}{p_\mathrm{g}^2p_\mathrm{s}}\left(\frac{1+2(1-p_\mathrm{g})[1-(1-p_\mathrm{g})^{t_\mathrm{c}}]}{1+ \frac{2(1-p_\mathrm{g})}{p_\mathrm{g}}[1-(1-p_\mathrm{g})^{t_\mathrm{c}}]}\right).
\end{align}

To compute the expected Werner parameter we must find the hitting probabilities $\gamma_{-1,\hat{t}}$ of states $\hat{t}\in \mathcal{S}_\mathrm{absorbing}$. The hitting probabilities satisfy the system \cref{eq: linear system werner parameter deterministic}.
Recursively substituting $\gamma_{t+1,\hat{t}}$ into $\gamma_{t,\hat{t}}$ we find that
\begin{align*}
    \gamma_{-1,\hat t} &= P_{-1,\hat t}+P_{-1, -1}\gamma_{-1,\hat t}+P_{{-1},0} \gamma_{0,\hat t}\\
    &= \dots \\
    &= P_{-1,\hat t}+\sum_{t=0}^{t_\mathrm{c}-1}P_{t,\hat t}\prod_{t'=0}^tP_{t'-1,t'}\\
    &\qquad+\left( P_{-1,-1}+\sum_{t=0}^{t_\mathrm{c}-1}P_{t,-1}\prod_{t'=0}^tP_{t'-1,t'}\right)\gamma_{-1,\hat t}.
\end{align*}
This equation can be solved for $\gamma_{-1,\hat t}$.
Note that for the transitions into the end-to-end states only $P_{t-1,\hat t}$ is nonzero, because the end-to-end link $\hat t$ can only be obtained when the time step starts with a link of age $t-1$ on one of the two segments. For $t=0$ we thus find
\begin{align*}
    \gamma_{-1,\hat{0}}
    &=  \frac{1}{1+ 2\frac{1-p_\mathrm{g}}{p_\mathrm{g}}[1-(1-p_\mathrm{g})^{t_\mathrm{c}}]},
\end{align*}
and for $t>0$, we find
\begin{align*}
    \gamma_{0,\hat t} =\frac{2(1-p_\mathrm{g})^{t}}{1+ 2\frac{1-p_\mathrm{g}}{p_\mathrm{g}}[1-(1-p_\mathrm{g})^{t_\mathrm{c}}]}.
\end{align*}
The expected Werner parameter follows from \cref{eq: werner deterministic} 
and can be written as
\begin{equation}
    \overline{w} = w_0^2 \frac{1+2(1-p_\mathrm{g})\lambda\frac{1-(1-p_\mathrm{g})^{t_\mathrm{c}}\lambda^{t_\mathrm{c}}}{1-(1-p_\mathrm{g})\lambda}}{1+ 2\frac{1-p_\mathrm{g}}{p_\mathrm{g}}[1-(1-p_\mathrm{g})^{t_\mathrm{c}}]}.
\end{equation}

\begin{figure}
    \centering
    \includegraphics[width=0.9\linewidth]{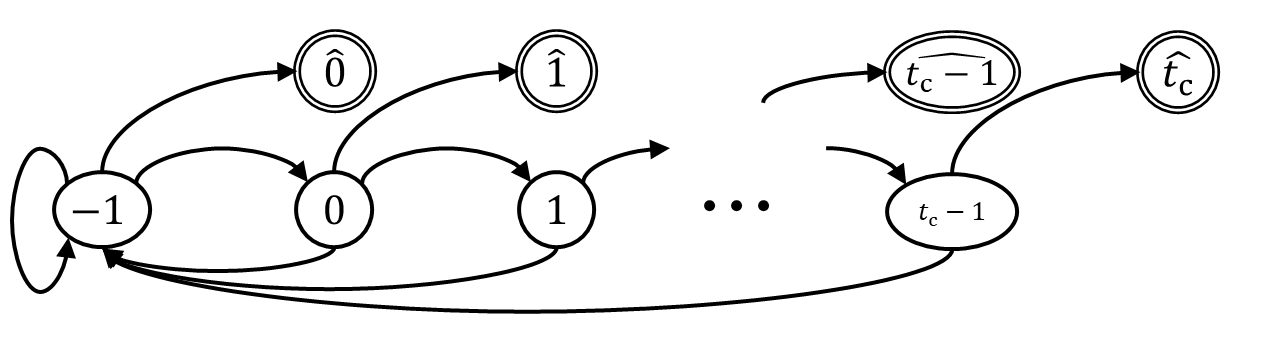}
    \caption{\textbf{Markov chain for three-node quantum repeater chain with deterministic cutoffs and cutoff time $t_\mathrm{c}$.} Transient states are circled with one line and absorbing states with double lines. Nonzero transition probabilities are indicated with arrows. The unit transition probability from an absorbing state back to itself is omitted.}
    \label{fig: Markov chain}
\end{figure}

\subsection{Probabilistic cutoffs}
An immediate generalization of the Markov chain model for the deterministic cutoff policy to the probabilistic cutoff policy would result in infinite state spaces as links can in principle reach arbitrarily large ages when probabilistic cutoffs are used. In turn, this would lead to infinite dimensional linear systems to solve, which is numerically intractable. In this section we therefore introduce a different Markov chain model, with finite state space, to compute the expected delivery time and Werner parameter for the case of probabilistic cutoffs. The details are worked out and implemented for chains of three, four and five nodes. The three node case is solved analytically. The extension to an arbitrary number of nodes is left to future work.

The states for the Markov chain model for the repeater chain with probabilistic cutoff policy are chosen to encode only the presence or absence of links on the various segments. As in the treatment of the deterministic cutoff policy we let $n=n_\mathrm{node}-1$ denote the number of segments of the repeater chain again. For the probabilistic cutoff policy we then define the state space $\mathcal{S}$ of the Markov chain to consist of all binary strings of length $n$, i.e. $\mathcal{S}=\{0,1\}^n$. Each binary string encodes the presence and absence of links on the various segments. A $0$ correspond to a segment with no link, while $k$ adjacent $1$'s correspond to an entangled link of length $k$. For example, in a five-node chain there are four segments and the string $1101$ describes an entangled link spanning the two left-most segments and an elementary link on the right-most segment. The empty chain corresponds to the all-zeros string $\vec 0 = 00\cdots0$. The string of all ones $\vec 1=11\cdots 1$ corresponds to the repeater chain state with the end-to-end link. The partition of the state space into transient and absorbing states is
\begin{align}
    \mathcal{S}_\mathrm{transient} &= \{0,1\}^n\setminus\{11\cdots 1\} \\
    \mathcal{S}_\mathrm{absorbing} &= \{11\cdots 1\}.
\end{align}
The transition probability matrix $P$ can be inferred from the time step by considering the possible outcomes of the HEG, swap, and cutoff phases. We work out the three-node case as an example at the end of this section.

\paragraph{Expected delivery time}
The computation of the expected delivery time of an end-to-end link can be done with the same method as used for the deterministic cutoff policy above. We compute the expected delivery time of the end-to-end link from the equations for the expected hitting times $v_s$ for $s\in \mathcal{S}_\mathrm{transient},$ where $v_s$ is the expected time to reach $\mathcal{S}_\mathrm{absorbing}$ when starting in state $s\in \mathcal{S}_\mathrm{transient}$. These expected hitting times satisfy
\begin{equation}\label{eq: linear system delivery time probabilistic}
    v_s = 1 + \sum_{s'\in \mathcal{S}_\mathrm{transient}} P_{s,s'}v_{s'}.
\end{equation}
This is a linear system of dimension $|S_\mathrm{transient}|=2^n-1$. The expected delivery time of an end-to-end link starting from the empty chain is found as $v_{\vec 0}=\overline{T}$. 

\paragraph{Expected Werner parameter}
The computation of the expected Werner parameter requires a different method from what was done for deterministic cutoffs since in the Markov chain model for the probabilistic cutoff policy we do not track the ages of links in the state space. The idea of this different method is to define a \textit{Werner vector} that contains all the Werner parameters of the links in the chain. As the chain evolves from state $s$ to $s'$ in a single time step this Werner vector is updated by multiplying it with a \textit{Werner update matrix} $M_{\{s,s'\}}$. The updates in consecutive time steps are obtained by taking products of the Werner update matrices, and weighing them properly by the transition probabilities. For each state $s\in \mathcal{S}_\mathrm{transient}$, there will then be an \textit{expected Werner update matrix} $\overline{M}_{\{s,\vec 1\}}$ that describes the expected update to the Werner vector as an end-to-end link is produced starting from state $s$. From the expected Werner update matrices we can compute the expected Werner parameter of the end-to-end link. The main result that we prove in this section is that the expected Werner update matrix can be found as the solution to a finite dimensional linear system of equations. This is the content of the following proposition.
\begin{proposition}\label{prop: linear sytem for expected werner matrix}
    The expected Werner update matrices $\overline{M}_{\{s,\vec 1\}}$ for $s\in \mathcal{S}_\mathrm{transient}$ satisfy the linear system of equations
    \begin{equation}\label{eq: linear system expected werner matrices}
        \overline{M}_{\{s,\vec 1\}} = H_{\{s,\vec 1\}} + \sum_{s'\in \mathcal{S}_\mathrm{transient}}H_{\{s,s'\}} \overline{M}_{\{s',\vec 1\}},
    \end{equation}
    where $H_{\{s,s'\}}:= P_{s,s'} M_{\{s,s'\}}$ are the \emph{one-step Werner update matrices.}
\end{proposition}

To prove this proposition we need to make the definitions of the Werner vector and its update matrices precise. We do this for chains of $n=2,3,4$ segments. 

\textbf{Werner vector.} For these chains we define the \emph{Werner vector} $$\uvec = (u_0,\dots,u_{n-1},u_n),$$ where the entry $u_a$ for indices $a\in\{0,\dots,n-1\}$ is the Werner parameter of the link on segment $a$. In particular, if an entangled link spans multiple segments $a,\dots,a+k$, then all $u_a, \dots, u_{a+k}$ are equal to the Werner parameter of this entangled link. If there is no link on segment $a$, then $u_a=1$. The final entry $u_n$ is defined to be the product of the Werner parameters of all links. For example, the Werner vector for the state $1101$ of a five node chain is of the form
$$\uvec = (
    w_0^2 \lambda^{t_0}, w_0^2 \lambda^{t_0}, 1, w_0 \lambda^{t_3}, w_0^3\lambda^{t_0+t_3}
),$$
where $w_0$ is the initial Werner parameter of a freshly generated link, $\lambda = e^{-1/\tau_\mathrm{coh}}$ is the depolarizing parameter, $t_0$ is the age of the link over segments $0$ and $1$, and $t_3$ is the age of the link on segment $3$.

\textbf{Updating the Werner vector.} With this definition, the Werner vector $\uvec'$ at the end of a time step can be deterministically expressed in terms of the Werner vector $\uvec$ at the start of the time step, as long as the starting state $s$ and ending state $s'$ are given. For chains of two, three and four segments, the entry of the Werner vector $\uvec'$ at index $a'$ can be given in terms of the Werner vector $\uvec$ by 
\begin{equation}\label{eq: werner entry update}
u'_{a'} = u_a w_0^{\ell}\lambda^{m}
\end{equation}
for some index $a\in \{0,\dots,n\}$ and integers $\ell,m\geq 0$.
The integer $\ell$ in \cref{eq: werner entry update} is the number of new elementary links that were used to construct the link in state $s'$, while the integer $m$ is the number of already-active links that went into the construction of the link in state $s'$. This is because any link in the state $s'$ is obtained by generating new elementary links and swapping them with links that were already active in state $s$. Moreover, in a chain of at most four segments there can be at most two disjoint links in any state $s$. The Werner vector is so constructed that it contains the Werner parameters of these two disjoint links individually, in the entries $u_a$ for $a\in \{0,\dots,n-1\}$, but it also contains the product of their Werner parameters in the entry $u_n$. We can thus find $\ell$, $m$ and $a$ in \cref{eq: werner entry update} by considering which links already present in state $s$ and which newly generated links in the transition from $s$ to $s'$ were used to construct the link with Werner parameter $u'_{a'}$ in state $s'$.
We describe the selection of $\ell$, $m$ and $a$ for a given $a'$ in more detail below. 

We first consider the case that a link is present at index $a'\in \{0,1,\dots,n\}$ in state $s'$. We subdivide this case further into the number of already-active links of state $s$ that end up in this link. In all these subcases, the integer $\ell$ must be chosen to correspond to the number of newly generated elementary links involved in constructing the link at index $a'$. 
\begin{itemize}
    \item \textbf{No already-active links in $a'$.} If the link corresponding to index $a'$ in state $s'$ does not involve any active links from state $s$, then $m=0$. The index $a$ may be taken equal to $a'$ so that $u_{a}$ equals $1$. 
    
    \item \textbf{One already-active link in $a'$.} If the link in state $s'$ involves one link that was already active in state $s$, then $m=1$. The index $a$ must be chosen so that $u_a$ equals the Werner parameter of that already-active link. 
    
    \item \textbf{Two already-active links in $a'$.} If the link in state $s'$ involves two disjoint links that were already active in state $s$, then $m=2$. The index $a$ must be chosen to be $n$ so that $u_a$ equals the product of the Werner parameters of the two already-active links.
\end{itemize}
In all these cases, the entry $u_n$ that keeps track of the product of the Werner parameters of all the links in the chain must be multiplied by factors of $w_0$ and $\lambda$ that correspond to all the newly generated links in the chain and all the links already active in state $s$ to give $u'_n$.

The above shows how the entries of the Werner vector $\uvec$ in state $s$ can be used to obtain entries of the Werner vector $\uvec'$ to describe the links in state $s'$. It can also happen that links of state $s$ disappear in the transition to the new state $s'$, either through failed swaps or through cutoffs. The entries of the Werner vector that correspond to the links that disappear must then get reset to $1$ so that they encode for empty segments again. This resetting can be done by setting $u'_{a'}=u_a$, where $a$ is the index of a segment that did not have a link in state $s$ (since then $u_a=1$). Such an $a$ must always exist in a state $s$ that has no end-to-end entanglement. 
The entry $u_n$ must be appropriately updated as well, there are three cases:
\begin{itemize}
    \item \textbf{No links in $s'$.} If no links remain in state $s'$ after failed swaps or cutoffs, then $u_n$ must be reset to $1$. This can be done by setting $u'_n=u_a$ for the same $a$ as above.

    \item \textbf{One link in $s'$.}  If one link remains in state $s'$ after failed swaps or cutoffs, then the entry $u_n$ must be set to the Werner parameter of the remaining link. This can be done by setting $u'_n=u_a w_0^\ell \lambda^{m}$, where $u_a$ is an entry that corresponds to the link that remains and appropriate factors of $w_0$ and $\lambda$ must be used to account for newly generated elementary links and aging already-active links. 

    \item \textbf{Two links in $s'$.} The case that two disjoint links remain in state $s'$ after failed swaps or cutoffs cannot occur in chains of at most four segments. This is because there can be at most two disjoint links to start with in such chains, and failed swaps or cutoffs strictly reduce the number of disjoint links.
\end{itemize}

\textbf{Werner update matrix.} Having established that for a given transition of states $s\to s'$ the update $\uvec \to \uvec'$ of the Werner vector is given by \cref{eq: werner entry update} we can now define the Werner update matrices.  We define the \emph{Werner update matrix} $(M_{\{s,s'\}})_{aa'}$ such that the column with index $a'$ has a single nonzero entry $w_0^\ell\lambda^m$ in row $a$, where $\ell$, $m$, and $a$ are as in \cref{eq: werner entry update}. With this definition, the update of the Werner vector can be represented by multiplying $\uvec$ from the right with $M_{\{s,s'\}}$. That is,
\begin{equation}\label{eq:update_rule_w_dep_M}
    u'_{a'} = \sum_{a=0}^n u_{a} (M_{\{s,s'\}})_{aa'}.
\end{equation}
All Werner update matrices for the three-node chain and some illustrative examples for the five-node chain are provided at the end of this section. All the Werner update matrices for four- and five-node chains can be found in the repository at \cite{Grimbergen2026}.

\textbf{One-step Werner update matrix.} Recall that our goal is to compute the expected Werner parameter of the end-to-end link starting from the empty state. The expected Werner parameter of the end-to-end link can be obtained by averaging over all the possible ways in which the empty state can evolve into the end-to-end state. This evolution follows the dynamics of the transition probability matrix $P$. To account for this when updating the Werner vectors we introduce the \textit{one-step Werner update matrix} $H$, defined by
\begin{equation}\label{eq: definition of H}
    H_{\{s,s'\}} = P_{s,s'} M_{\{s,s'\}}.
\end{equation}
In the one-step Werner update matrix all Werner update matrices $M_{\{s,s'\}}$ are weighed by their corresponding transition probability. 

\begin{remark}
More generally, the \emph{$k$-step Werner update matrix} for $k\geq 0$, is defined as the  $k$-th power of $H$. For $k\geq 2$ it can be computed as
\begin{multline}
    H^k_{\{s_0,s_k\}} = \sum_{s_1,\dots,s_{k-1} \in \mathcal{S}_\mathrm{transient}}P_{s_0,s_1}\cdots P_{s_{k-1},s_k}\cdot\\ M_{\{s_0,s_1\}}\cdots M_{\{s_{k-1},s_{k}\}}.
\end{multline}
For $k=0$ we define $H^0_{\{s,s'\}} = \delta_{s,s'} \II_{n+1}$, where $\II_{n+1}$ is the identity matrix in Werner vector space.
\end{remark}
Hence, $H^k_{\{s_0,s_k\}}$ sums together the products of the Werner update matrices over all paths of length $k$ from a state $s_0$ to a state $s_k$, weighed by the probability that that path occurs. For a path of zero steps $H^0$ acts as the identity on both the state and the Werner vector.

\textbf{Expected Werner update matrix.} For a transient state $s \in \mathcal{S}_\mathrm{transient}$, we then define the \textit{expected Werner update matrix} $\overline{M}_{\{s,\vec 1\}}$ as
\begin{equation}\label{eq: average werner update matrix infinite sum}
    \overline{M}_{\{s,\vec 1\}} = \sum_{k=0}^\infty\sum_{s'\in \mathcal{S}_\mathrm{transient}} H^k_{\{s,s'\}} H_{\{s',\vec 1\}}.
\end{equation}
For each $k\geq 0$ and $s'\in \mathrm{S}_\mathrm{transient}$ the term $H^k_{\{s,s'\}}H_{\{s',\vec 1\}}$ averages the updates to the Werner vector over paths of length $k+1$ that start from $s$ and enter the end-to-end state $\vec 1$ from state $s'$. The sums over $k$ and $s'$ ensure that all possible paths from the empty chain to the end-to-end state are taken into account. \cref{lemma: overline M well-defined}  below shows that the series in \cref{eq: average werner update matrix infinite sum} converges so that the expected Werner update matrix is well-defined.

\textbf{Expected Werner parameter.} The \textit{expected Werner parameter} $\overline{w}$ of the end-to-end link can be computed by taking the $n$-th component of the Werner vector after applying the expected Werner update matrix $\overline{M}_{\{\vec 0,\vec 1\}}$ to the Werner vector $\onesvec=(1,\dots,1)$ of the empty chain $s=\vec 0$. In other words, \begin{equation}\label{eq: expected werner probabilistic}
    \overline{w} = \sum_{a=0}^n(\overline{M}_{\{\vec 0, \vec 1\}})_{an}.
\end{equation}

To compute $\overline{w}$ we must therefore compute the matrix $\overline{M}_{\{\vec 0, \vec 1\}}$.
In \cref{eq: average werner update matrix infinite sum} we have an expression for this matrix in terms of an infinite series, but we cannot evaluate this analytically, and neither can we evaluate it exactly on a computer. However, as stated in \cref{prop: linear sytem for expected werner matrix} above we can rewrite \cref{eq: average werner update matrix infinite sum} into a linear system for all $\overline{M}_{\{s,\vec 1\}}$ with $s\in \mathcal{S}_\mathrm{transient}$ that can be solved analytically, or exactly on a computer. We now prove this proposition.
\begin{proof}[\textbf{Proof of \cref{prop: linear sytem for expected werner matrix}}]
    The proof is similar to that of \cite[Theorems 1.42 and 1.43]{serfozo2009basics} which gives rise the linear systems in \cref{eq: linear system delivery time deterministic}, \cref{eq: linear system werner parameter deterministic} and \cref{eq: linear system delivery time probabilistic} that have been presented above. Starting from \cref{eq: average werner update matrix infinite sum} we compute
    \begin{align*}
        \overline{M}_{\{s,\vec 1\}} &= \sum_{k=0}^\infty\sum_{s''\in \mathcal{S}_\mathrm{transient}} H^k_{\{s,s''\}} H_{\{s'',\vec 1\}} \\
        &= H_{\{s,\vec 1\}} + \sum_{k=1}^\infty\sum_{s''\in \mathcal{S}_\mathrm{transient}} H^k_{\{s,s''\}} H_{\{s'',\vec 1\}} \\
        &= H_{\{s,\vec 1\}} + \sum_{k=1}^\infty\sum_{s',s''\in \mathcal{S}_\mathrm{transient}} H_{\{s,s'\}}H^{k-1}_{\{s',s''\}} H_{\{s'',\vec 1\}} \\
        &= H_{\{s,\vec 1\}} \\
        &\,+ \sum_{s'\in \mathcal{S}_\mathrm{transient}}\sum_{k=1}^\infty\sum_{s''\in \mathcal{S}_\mathrm{transient}} H_{\{s,s'\}}H^{k-1}_{\{s',s''\}} H_{\{s'',\vec 1\}} \\
        &= H_{\{s,\vec 1\}} + \sum_{s'\in \mathcal{S}_\mathrm{transient}} H_{\{s,s'\}}\overline{M}_{\{s',\vec 1\}}.
    \end{align*}
    The first equality is the definition of $\overline{M}_{\{s,\vec 1\}}$ as given in \cref{eq: average werner update matrix infinite sum} and the second equality uses the definition of $H^0$ to split the sum in two parts. The third equality uses the definition of $H^k$. The series over $k$ and the sum over $s'$ may be interchanged in the fourth step because \cref{lemma: overline M well-defined} shows that the series converges. The final step uses \cref{eq: average werner update matrix infinite sum} again, after reindexing the series. This shows that the $\overline{M}_{\{s,\vec 1\}}$ for $s\in \mathcal{S}_\mathrm{transient}$ satisfy the linear system \cref{eq: linear system expected werner matrices} and thereby proves the proposition.
\end{proof}

\cref{prop: linear sytem for expected werner matrix} gives a linear system for the expected Werner matrices. This system can solved analytically (for three-node chains) or exactly on a computer. The expected Werner parameter can then be computed from \cref{eq: expected werner probabilistic}. For chains of four and five nodes we have built the state space, the transition probability matrix and the Werner update matrices in the code at \cite{Grimbergen2026}. The code then solves for the expected delivery time and expected Werner parameter. The dimensions of the linear system for four nodes is $28$ and for five nodes it is $75$. In both cases our numerical implementation using \texttt{numpy.linalg.solve} is sufficient. For three nodes we provide analytic solutions for the expected delivery time and expected Werner parameter below. But first we show that the expected Werner update matrix as defined by \cref{eq: average werner update matrix infinite sum} is well-defined.

\paragraph{The expected Werner update matrix is well-defined}
We show that the series that defines the expected Werner update matrix in \cref{eq: average werner update matrix infinite sum} is well-defined, i.e. that it converges. For this we need the following remark and lemma.
\begin{remark}\label{remark: structure of Ms}
    By construction, each $M_{\{s,s'\}}$ has a single nonzero entry in each column. Moreover, this nonzero entry takes values in $(0,1]$.
\end{remark} 
\begin{lemma}\label{lemma: H bounded by Q}
    The $k$-step Werner update matrix $H^k$ satisfies $$(H^k_{\{s,s'\}})_{aa'} \leq P^k_{s,s'}$$ for all $s,s'\in \mathcal{S}$, $a,a'\in \{0,\dots,n\}$, and $k\geq 0$, where $P^k$ is the $k$-th power of $P$.
\end{lemma}
\begin{proof}
    For $k=0$ the result follows by definition.
    For $k=1$ we have for all $s,s'\in \mathcal{S}$ and $a,a'\in \{0,\dots,n\}$ that, $$(H_{\{s,s'\}})_{aa'}=P_{s,s'}(M_{\{s,s'\}})_{aa'}\leq P_{s,s'},$$
    since $(M_{\{s,s'\}})_{aa'}\in [0,1]$ by \cref{remark: structure of Ms}.
    By induction we now assume that for some $k\geq 1$ the entry-wise inequality $(H^k_{\{s,s'\}})_{aa'} \leq P_{s,s'}^k$ holds for all $s,s'\in \mathcal{S}$ and $a,a'\in \{0,\dots,n\}$. Then, for $s,s''\in \mathcal{S}$, $a,a''\in \{0,\dots,n\}$ and $k+1$ we compute
    \begin{align*}
        (H^{k+1}_{\{s,s''\}})_{aa''} &= \sum_{s',a'}(H^{k}_{\{s,s'\}})_{aa'}(H_{\{s',s''\}})_{a'a''} \\
        &\leq \sum_{s',a'}P_{s,s'}^k P_{s',s''}(M_{\{s',s''\}})_{a'a''} \\
        &= \sum_{s'}P_{s,s'}^k P_{s's''}(M_{\{s',s''\}})_{\underline a' a''} \\
        &\leq \sum_{s'}P_{s,s'}^{k}P_{s',s''} \\
        &= P_{s,s''}^{k+1}.
    \end{align*}
    In the first step we have used the induction hypothesis,  $(H_{\{s,s'\}}^k)_{aa'}\leq P^k_{s,s'}$. In the second and third step we used \cref{remark: structure of Ms} to conclude that $\sum_{a'} (M_{\{s',s''\}})_{a'a''}=(M_{\{s',s''\}})_{\underline a' a''}\leq 1$, where $\underline a'$ is the unique positive element in column $a''$ of $M_{\{s',s''\}}$.
    This completes the induction step, and the Lemma is proven.
\end{proof}
We use \cref{lemma: H bounded by Q} to show that the series that defines the expected Werner update matrix converges.
\begin{lemma}\label{lemma: overline M well-defined}
    The series $\sum_{k=0}^\infty \sum_{s'\in \mathcal{S}_\mathrm{transient}} H^k_{\{s,s'\}} H_{\{s,\vec 1\}}$ in \cref{eq: average werner update matrix infinite sum} converges.
\end{lemma}
\begin{proof}
    Since all entries of $H$ are positive we can show that the series $\sum_{k=0}^\infty \sum_{s'\in \mathcal{S}_\mathrm{transient}} H^k_{\{s,s'\}} H_{\{s,\vec 1\}}$ converges by providing an upperbound.
    First of all, for any pair of indices $a,b\in \{0,\dots,n\}$ we have that
    \begin{align*}
        &\sum_{k=0}^\infty\sum_{s'\in \mathcal{S}_\mathrm{transient}} (H^k_{\{s,s'\}} H_{\{s',\vec 1\}})_{ab} = \\
        &\qquad\qquad\sum_{k=0}^\infty\sum_{s'\in \mathcal{S}_\mathrm{transient}} \sum_{a'}(H^k_{\{s,s'\}})_{aa'} (H_{\{s',\vec 1\}})_{a'b}\\
        &\qquad\qquad\leq \sum_{k=0}^\infty \sum_{s'\in \mathcal{S}_\mathrm{transient}}(H^k_{\{s,s'\}})_{a\underline a'},
    \end{align*}
    where $\underline a'$ is the index corresponding to the only nonzero element in column $b$ of $H_{\{s',\vec 1\}}$. The inequality follows because  $(H_{\{s',\vec 1\}})_{\underline a'b}= P_{s',\vec 1} (M_{\{s',\vec 1\}})_{\underline a' b}\leq 1$. From \cref{lemma: H bounded by Q} it follows that
    \begin{equation*}
        \sum_{k=0}^\infty\sum_{s'\in \mathcal{S}_\mathrm{transient}} (H^k_{\{s,s'\}})_{a\underline a'} \leq \sum_{k=0}^\infty \sum_{s'\in \mathcal{S}_\mathrm{transient}} P_{s,s'}^k. 
    \end{equation*}
    Since both $s$ and $s'$ are transient states in this expression, we have that $P^k$ is the $k$-th power of the transition probability matrix restricted to the transient states. Therefore, $\sum_{s'\in \mathcal{S}_\mathrm{transient}} P^k_{s,s'}$ is equal to the probability that it takes more than $k$ time steps to reach $\mathcal{S}_\mathrm{absorbing}$ from state $s$.
    If we denote by $T_s$ the hitting time of $\mathcal{S}_\mathrm{absorbing}$ starting from state $s$, then we have that
    \begin{align*}
        \sum_{k=0}^\infty \sum_{s'\in \mathcal{S}_\mathrm{transient}}P^k_{s,s'} &= \sum_{k=0}^\infty \Pr[T_s>k] = \EE[T_s].
    \end{align*}
    The expected hitting time $\EE[T_s]$ is just the expected delivery time $v_s$ from \cref{eq: linear system delivery time probabilistic}.
    Therefore, we conclude that the series in \cref{eq: average werner update matrix infinite sum} gives a matrix for which each entry is upperbounded by the expected delivery time of the end-to-end link.
    This expected delivery time is finite because it in turn is upperbounded by the expectation value of a geometric random variable with success probability $p_\mathrm{g}^np_\mathrm{s}^{n-1}$. Indeed, the probability to generate the end-to-end link in any time step is at least $p_\mathrm{g}^np_\mathrm{s}^{n-1}$, which is the probability that all elementary link generations and swaps succeed simultaneously. We have thus found a finite upperbound on the series. 
\end{proof}

\paragraph{Three-node chains}
We employ the methods introduced in the paragraphs above to analytically compute the expected delivery time and expected Werner parameter for three-node chains with the probabilistic cutoff policy. 
In the case of a three-node chain we get the state spaces $\mathcal{S}_\mathrm{transient} = \{00, 10,01\}$ and $\mathcal{S}_\mathrm{absorbing} = \{11\}$. Using the assumption of the homogeneity of the system we can equivalently work with the state spaces $\mathcal{S}_\mathrm{transient}=\{00,10\}=\{0,1\}$ and $\mathcal{S}_\mathrm{absorbing}=\{11\}=\{2\}$, where we identify $0\equiv 00$, $1\equiv 10\equiv 01$ and $2\equiv 11$. As we explain below, the nonzero transition probabilities out of state $0$ are given by
\begin{align*}
    P_{0,1} &= 2p_\mathrm{g}(1-p_\mathrm{g})(1-p_\mathrm{c}) \\
    P_{0,2} &= p_\mathrm{g}^2 p_\mathrm{s} \\
    P_{0,0} &= (1-p_\mathrm{g})^2 + 2p_\mathrm{g}(1-p_\mathrm{g})p_\mathrm{c} + p_\mathrm{g}^2(1-p_\mathrm{s})
\end{align*}
The probability $P_{0,1}$ is the probability to transition from the empty state to a state with exactly one link at the end of the time step. This can only happen if a single link is generated and is not subsequently discarded on exactly one of the two segments. The probability $P_{0,2}$ is the probability that the end-to-end link is generated from the empty chain in a single time step. For this to happen both links must be generated and successfully swapped together. Finally, the probability $P_{0,0}$ is the probability that the chain returns to the empty state at the end of the time step. This can happen in three ways. Either no links are generated, or one link is generated but discarded, or two links are generated but the swap fails. Similar reasoning shows that the nonzero transition probabilities out of state $1$ are given by
\begin{align*}
    P_{1,2} &= p_\mathrm{g}p_\mathrm{s} \\
    P_{1,1} &= (1-p_\mathrm{g})(1-p_\mathrm{c}) \\
    P_{1,0} &=  p_\mathrm{g}(1-p_\mathrm{s})+(1-p_\mathrm{g})p_\mathrm{c}.
\end{align*}

The expected delivery time of an end-to-end link starting from the empty repeater chain is given by $v_0$, where the vector $(v_0,v_1)$ is the solution to \cref{eq: linear system delivery time probabilistic}. In terms of the transition probabilities we can write the linear system for $(v_0,v_1)$ in matrix form as
\begin{equation}
    \begin{pmatrix}
        1 - P_{0,0} & -P_{0, 1} \\ -P_{1,0} & 1- P_{1,1}
    \end{pmatrix}
    \begin{pmatrix}
        v_0 \\ v_1
    \end{pmatrix}
    = \begin{pmatrix}
        1 \\ 1
    \end{pmatrix}.
\end{equation}
The solution is
\begin{align*}
    \begin{pmatrix}
        v_0 \\ v_1
    \end{pmatrix}
    &= \begin{pmatrix}
        1 - P_{0,0} & -P_{0, 1} \\ -P_{1,0} & 1- P_{1,1}
    \end{pmatrix}^{-1}
    \begin{pmatrix}
        1 \\ 1
    \end{pmatrix}.
\end{align*}
Plugging in the transition probabilities in terms of $p_\mathrm{g}$, $p_\mathrm{s}$ and $p_\mathrm{c}$ this yields
\begin{align}
    \overline{T} &= \frac{1}{p_\mathrm{g}^2p_\mathrm{s}}\left(\frac{1+2p_\mathrm{g}\frac{(1-p_\mathrm{g})(1-p_\mathrm{c})}{1-(1-p_\mathrm{g})(1-p_\mathrm{c})}}{1+2\frac{(1-p_\mathrm{g})(1-p_\mathrm{c})}{1-(1-p_\mathrm{g})(1-p_\mathrm{c})}} \right).
\end{align}

To compute the expected Werner parameter we must first give the Werner update matrices $M_{\{s,s'\}}$. With the convention that state $1$ corresponds to a link on the left segment, i.e. $1\equiv10$, the Werner update matrices can be given by
\begin{align*}
    M_{\{0,0\}} &= \begin{pmatrix}
        1 & 0 & 0 \\ 0 & 1 & 0 \\ 0 & 0 & 1
    \end{pmatrix} \\
    M_{\{0,1\}} &= \begin{pmatrix}
        w_0 & 0 & 0 \\ 0 & 1 & 0 \\ 0 & 0 & w_0
    \end{pmatrix} \\
    M_{\{0,2\}} &= \begin{pmatrix}
        0 & 0 & 0 \\ 0 & 0 & 0 \\ w_0^2 & w_0^2 & w_0^2
    \end{pmatrix} \\
    M_{\{1,0\}} &= \begin{pmatrix}
        0 & 0 & 0 \\ 1 & 1 & 1 \\ 0 & 0 & 0
    \end{pmatrix} \\
     M_{\{1,1\}} &= \begin{pmatrix}
        \lambda & 0 & 0 \\ 0 & 1 & 0 \\ 0 & 0 & \lambda
    \end{pmatrix} \\
     M_{\{1,2\}} &= \begin{pmatrix}
        0 & 0  & 0 \\ 0 & 0 & 0 \\ w_0\lambda & w_0\lambda & w_0\lambda 
    \end{pmatrix}.
\end{align*}
Note that the Werner update matrix $M_{\{1,0\}}$ resets all entries of the Werner vector $\uvec$ to $1$ because in the convention that the link always sits on the left segment we have that $u_1=1$. Therefore,
$$\uvec M_{\{1,0\}} = \begin{pmatrix}
    u_0 &u_1 &u_2
\end{pmatrix}\begin{pmatrix}
        0 & 0 & 0 \\ 1 & 1 & 1 \\ 0 & 0 & 0
    \end{pmatrix} = \begin{pmatrix}
    1 &1 &1
\end{pmatrix},$$
which indeed properly resets the Werner vector to that of the empty state.

The expected Werner parameter follows from \cref{eq: expected werner probabilistic} and requires the computation of $(\overline{M}_{\{0,2\}})_{an}$ for $a\in \{0,1,2\}$. This computation is done by solving the linear system \cref{eq: linear system expected werner matrices} for $(\overline{M}_{\{s,2\}})_{an}$ with $s\in \{0,1\}$ and $a\in \{0,1,2\}.$ To write the linear system in matrix form we consider the $(sa)$-basis $\{(00),(01),(02),(10),(11),(12)\}$. We define the $(sa)$-indexed vector $\overline{\Mvec}$ through
\begin{equation}
    \overline{M}_{(sa)} = (\overline{M}_{\{s,2\}})_{a2},
\end{equation}
which contains the relevant components of the expected Werner matrix,
and similarly we define the $(sa)$-indexed vector $\Hvec$ through
\begin{equation}
    H_{(sa)} = (H_{\{s,2\}})_{a2},
\end{equation}
which contains the relevant components of the final Werner update step. Explicitly, we have that the vector $\Hvec$ in the $(sa)$-basis is given by
\begin{equation}
    \Hvec = \begin{pmatrix}
        0 \\ 0 \\ P_{0,2}w_0^2 \\ 0 \\ 0 \\ P_{1,2}w_0 \lambda 
    \end{pmatrix}.
\end{equation}
We can then write the linear system in \cref{prop: linear sytem for expected werner matrix} for the relevant components of the expected Werner update matrix as
\begin{equation}\label{eq: three node episodic probabilistic h method equation}
    (\II_6 - H)\overline{\Mvec} = \Hvec,
\end{equation}
where the matrix $H$ in the $(sa)$-basis is
\begin{equation}
    H = \begin{pmatrix}
        P_{0,0} & 0 & 0 & P_{0,1} w_0 & 0 & 0 \\
        0 & P_{0,0} & 0 & 0 & P_{0,1} & 0 \\
        0 & 0 & P_{0,0} & 0 & 0 & P_{0,1}w_0 \\
        0 & 0 & 0 & P_{1,1}\lambda & 0 & 0 \\
        P_{1,0} & P_{1,0} & P_{1,0} & 0 & P_{1,1} & 0 \\
        0 & 0 & 0 & 0 & 0 & P_{1,1}\lambda.
    \end{pmatrix}
\end{equation}
Solving \cref{eq: three node episodic probabilistic h method equation} we find that the expected Werner parameter is given by
\begin{align}
    \overline{w}
    &= w_0^2 \left( \frac{1 + 2 \frac{(1-p_\mathrm{g})(1-p_\mathrm{c})\lambda}{1-(1-p_\mathrm{g})(1-p_\mathrm{c})\lambda}}{1 + 2 \frac{(1-p_\mathrm{g})(1-p_\mathrm{c})}{1-(1-p_\mathrm{g})(1-p_\mathrm{c})}} \right).
\end{align}

Let us end this example of the three-node chain by noting that the analytic expressions for the expected delivery time and the expected Werner parameter in the three-node chain have a structure to them that is the same in both the deterministic and probabilistic cutoff policy. To make this manifest, we define two functions $A_\lambda(t_\mathrm{c})$ and $B_\lambda(p_\mathrm{c})$ as
\begin{align}
    A_\lambda(t_\mathrm{c}) &= \sum_{t=1}^{t_\mathrm{c}}(1-p_\mathrm{g})^t \lambda^t \label{eq: A lambda sum} \\
    &= [1-(1-p_\mathrm{g})^{t_\mathrm{c}}\lambda^{t_\mathrm{c}}]\frac{(1-p_\mathrm{g})\lambda }{1-(1-p_\mathrm{g})\lambda} \label{eq: A lambda}
\end{align}
and
\begin{align}
    B_\lambda(p_\mathrm{c}) &= \sum_{t=1}^\infty (1-p_\mathrm{g})^t(1-p_\mathrm{c})^t\lambda^t \label{eq: B lambda sum}
    \\
    &=  \frac{(1-p_\mathrm{g})(1-p_\mathrm{c})\lambda}{1-(1-p_\mathrm{g})(1-p_\mathrm{c})\lambda}\label{eq: B lambda},
\end{align}
where $\lambda = e^{-1/\tau_\mathrm{coh}}$ is the depolarizing parameter. The expected delivery time and expected Werner parameter of the deterministic cutoff policy can be expressed in terms of $A_\lambda(t_\mathrm{c})$ as 
\begin{align}\label{eq: analytic three node deterministic appendix}
    \overline{T} = \frac{1}{p_\mathrm{g}^2p_\mathrm{s}}\frac{1+2p_\mathrm{g}A_1(t_\mathrm{c})}{1+2A_1(t_\mathrm{c})} &&\mathrm{and}&& \overline{w}= w_0^2\frac{1+2A_\lambda(t_\mathrm{c})}{1+2A_1(t_\mathrm{c})}.
\end{align}
Similarly, the expected delivery time and expected Werner parameter of the probabilistic cutoff policy can be expressed in terms of $B_\lambda(p_\mathrm{c})$ as 
\begin{align}\label{eq: analytic three node probabilistic appendix}
    \overline{T} = \frac{1}{p_\mathrm{g}^2p_\mathrm{s}}\frac{1+2p_\mathrm{g}B_1(p_\mathrm{c})}{1+2B_1(p_\mathrm{c})} &&\mathrm{and}&& \overline{w}= w_0^2\frac{1+2B_\lambda(p_\mathrm{c})}{1+2B_1(p_\mathrm{c})}.
\end{align}
We will use these expression in \cref{appendix: RF supplement} to show that for the three-node chain the probabilistic cutoff always has a lower fidelity than the deterministic cutoff policy when compared at the same rate.

\paragraph{More examples of Werner update matrices for five-node chains}
To illustrate the construction of the Werner update matrices more concretely we provide examples below for transitions in a five-node chain. Consider the transition from state $s=0100$ to state $s'=1101$. The general form of the Werner vector in state $s=0100$ is
$\uvec=(1,w_0\lambda^{t_1},1,1,w_0\lambda^{t_1})$, where $t_1$ is the age of the link on segment $1$. Upon the transition to state $s'=1101$ the Werner vector must be updated according to
\begin{multline*}
\uvec =(
    1, w_0\lambda^{t_1}, 1, 1, w_0\lambda^{t_1}
)\to \\
\uvec'=(
    w_0^2\lambda^{t_1+1}, w_0^2\lambda^{t_1+1}, 1, w_0, w_0^3\lambda^{t_1+1}
),
\end{multline*}
so that $u'_0=u'_1=w_0^2\lambda^{t_0+1}$ are equal to the Werner parameter of the link that resulted from the swap of the new elementary link on segment $0$ and the already-active elementary link on segment $1$, $u'_2=1$ corresponds to the empty segment $2$, $u'_3=w_0$ corresponds to the new elementary link on segment $3$, and $u_4=w_0^3\lambda^{t_1+1}$ is the product of the Werner parameters of all links in state $s'$. We see that
\begin{align*}
    u'_0 &= u_1 w_0\lambda \\
    u'_1 &= u_1 w_0\lambda \\
    u'_2 &= u_2 \\
    u'_3 &= u_3 w_0 \\
    u'_4 &= u_4 w_0^2\lambda.
\end{align*}
This update can be achieved through right multiplication of $\uvec$ by the matrix
$$
M_{\{0100,1101\}}=\begin{pmatrix}
    0 & 0 & 0 & 0 & 0\\
    w_0\lambda & w_0\lambda & 0 & 0 & 0 \\
    0 & 0 & 1 & 0 & 0 \\
    0 & 0 & 0 & w_0 & 0 \\
    0 & 0 & 0 & 0 & w_0^2\lambda \\
\end{pmatrix}.
$$

As a second example consider the transition $s=1010$ to $s'=1110$. The Werner vector update $\uvec$ to $\uvec'$ must be of the form
\begin{multline*}
(
    w_0\lambda^{t_0}, 1, w_0\lambda^{t_2}, 1, w_0^2\lambda^{t_0+t_2})\to \\ 
(
    w_0^3\lambda^{t_0+t_2+2}, w_0^3\lambda^{t_0+t_2+2}, w_0^3\lambda^{t_0+t_2+2}, 1, w_0^3\lambda^{t_0+t_2+2} 
),
\end{multline*}
where $t_0$ and $t_2$ are the ages of the already-active links in segments $0$ and $2$ in state $s=1010$. 
The matrix that describes this transition is
$$
M_{\{1010,1110\}}=\begin{pmatrix}
    0 & 0 & 0 & 0 & 0 \\
    0 & 0 & 0 & 0 & 0 \\
    0 & 0 & 0 & 0 & 0 \\
    0 & 0 & 0 & 1 & 0\\
    w_0\lambda^2 & w_0\lambda^2 & w_0\lambda^2 & 0 & w_0\lambda^2 \\
\end{pmatrix}.
$$

Transitions in which links disappear can also be formulated in terms of \cref{eq:update_rule_w_dep_M}. It does not matter whether the link disappears due to a failed swap or due to a cutoff, since the Werner update $\uvec\to\uvec'$ is the same in either case. When a link disappears, the Werner parameter of its segment must be reset to $1$. The crucial point is that as long as there is no end-to-end entanglement there is always at least one segment without a link, and thus at least one entry in $\uvec$ is equal to $1$. This entry can be used to reset other segments to empty. To illustrate this, consider the transition $s=1010$ to $s'=0010$. The Werner vector has to be updated according to
\begin{multline*}
\uvec = (
    w_0\lambda^{t_0}, 1, w_0\lambda^{t_2},1, w_0^2\lambda^{t_0+t_2}
)\to \\
\uvec'=(
    1 , 1 , w_0\lambda^{t_2+1} , 1 ,w_0 \lambda^{t_2+1}
).
\end{multline*}
Notice that $u'_4$ is the product of the Werner parameters of all active links in state $s'=0010$, which in this case is only the link on segment $2$.
A matrix that gives the correct update on the Werner vectors is
$$
M_{\{1010,0010\}}=\begin{pmatrix}
    0 & 0 & 0 & 0 & 0 \\
    1 & 1 & 0 & 0 & 0 \\
    0 & 0 & \lambda & 0 & \lambda \\
    0 & 0 & 0 & 1 & 0\\
    0 & 0 & 0 & 0 & 0 \\
\end{pmatrix}.
$$
Note that the above matrix uses $u_1$ to reset $u_0$ to $1$. We could have chosen to use $u_4$ instead. This choice is irrelevant.

\paragraph{Extension to longer chains}
For more than five nodes the dimensionality of the Werner vector should increase beyond $n+1$. This is because for more than five nodes there can be more than two disjoint links in a state $s$. For example, in a six-node chain the transition $s=10101$ to $s'=00101$ could occur. There is no way to get the new Werner parameter $u'_5$ here in our present formalism. It should equal $u_2u_4\lambda^2$ because this is the product of the Werner parameters of all the active links in the chain, but the product $u_2u_4$ is not contained in the $n+1$-dimensional Werner vector considered up to this point. The general strategy when going to longer chains should be to allocate an entry in the Werner parameter for every possible product of the Werner parameters of disjoint links. We leave such generalization to future work.

\section{Supplementary material: Rate and Fidelity}\label{appendix: RF supplement}
In \cref{section: rate and fidelity} we stated that the probabilistic cutoff policy yields a lower fidelity than the deterministic cutoff policy, when compared at equal rate. Intuitively, we explained this result from the fact that in the probabilistic cutoff policy arbitrarily old links are allowed to contribute to the rate while in the deterministic cutoff policy the age of the end-to-end link is strictly constrained by some finite maximum. The ensemble of end-to-end links that constitute $\rho_\mathrm{e2e}$ for some particular value of the rate should thus have a higher expected age in the case of the probabilistic cutoff policy. We now prove this claim rigorously for the case $n_\mathrm{node}=3$, and provide numerical evidence for $n_\mathrm{node}=4$ and $n_\mathrm{node}=5$.

In \cref{appendix: Markov chain models} we have shown that the rate and expected Werner parameter in the three-node chain can be expressed analytically in terms of functions $A_\lambda(t_\mathrm{c})$ and $B_\lambda(p_\mathrm{c})$. If we are given a cutoff time $t_\mathrm{c}$ and wish to compare the two cutoff policies at the same rate, or equivalently, at the same expected delivery time, we must find $p_\mathrm{c}^*$ such that $A_1(t_\mathrm{c}) = B_1(p_\mathrm{c}^*)$. Using the explicit expressions for $A_\lambda(t_\mathrm{c})$ and $B_\lambda(p_\mathrm{c})$ given in \cref{appendix: Markov chain models}, we find
\begin{equation}\label{eq:critical pc}
    p_\mathrm{c}^* = \frac{p_\mathrm{g}(1-p_\mathrm{g})^{t_\mathrm{c}}}{1-(1-p_\mathrm{g})^{t_\mathrm{c}+1}}.
\end{equation}

We will show that $F(t_\mathrm{c})\geq F(p_\mathrm{c}^*)$.
The fidelity $F$ is given in terms of the expected Werner parameter as $F=\flatfrac{(1+3\overline{w})}{4}$. Hence, if we show that $\overline{w}(t_\mathrm{c})\geq \overline{w}(p_\mathrm{c}^*)$ then we may conclude that $F(t_\mathrm{c})\geq F(p_\mathrm{c}^*)$ as well. It can be seen from \cref{eq: analytic three node deterministic appendix} and \cref{eq: analytic three node probabilistic appendix} that in order to show that $\overline{w}(t_\mathrm{c})\geq \overline{w}(p_\mathrm{c}^*)$ it suffices to show that
\begin{equation} 
A_\lambda(t_\mathrm{c})\geq B_\lambda(p_\mathrm{c}^*)
\end{equation}
for any value of $\lambda\in (0,1]$ because by definition $A_1(t_\mathrm{c})=B_1(p_\mathrm{c}^*)$.

Our strategy is as follows. First we find the requirement on any $p_\mathrm{c}$ to satisfy $A_{\lambda}(t_\mathrm{c})\geq B_{\lambda}(p_\mathrm{c})$ for fixed $t_\mathrm{c}$, and then we show that $p_\mathrm{c}^*$ satisfies this requirement. From the explicit expressions for $A_\lambda(t_\mathrm{c})$ and $B_\lambda(p_\mathrm{c})$ it can be shown that
\begin{equation}\label{eq:pc bound proof two segment}
A_\lambda(t_\mathrm{c})\geq B_{\lambda}(p_\mathrm{c}) \iff p_\mathrm{c}\geq \frac{(1-\nu_\lambda)\nu_\lambda^{t_\mathrm{c}}}{1-\nu_\lambda^{t_\mathrm{c}+1}},
\end{equation}
where we use the shorthand notation $\nu_\lambda=(1-p_\mathrm{g})\lambda$.

We now show that $p_\mathrm{c}^*$ satisfies the latter bound in \cref{eq:pc bound proof two segment}. Also using the shorthand notation $\nu_\lambda=(1-p_\mathrm{g})\lambda$ to express $p_\mathrm{c}^*$ as given by \cref{eq:critical pc}, we see that \cref{eq:pc bound proof two segment} for $p_\mathrm{c}= p_\mathrm{c}^*$ becomes
\begin{equation}\label{eq: final inequality proof two segment}
    \frac{(1-\nu_1)\nu_1^{t_\mathrm{c}}}{1-\nu_1^{t_\mathrm{c}+1}}\geq \frac{(1-\nu_\lambda)\nu_\lambda^{t_\mathrm{c}}}{1-\nu_\lambda^{t_\mathrm{c}+1}}.
\end{equation}
For $t_\mathrm{c}=0$ the inequality is satisfied because both sides evaluate to $1$. For $t_\mathrm{c}>0$ we note that the left hand side is simply the right hand side evaluated at $\lambda=1$. The inequality can thus be proven by showing that the right hand side is an increasing function in $\lambda$ for $\lambda \in (0,1]$. To see that this is so, we use the geometric series to rewrite the right hand side as
\begin{equation}
    \frac{(1-\nu_\lambda)\nu_\lambda^{t_\mathrm{c}}}{1-\nu_\lambda^{t_\mathrm{c}+1}} = \left[\sum_{t=0}^{t_\mathrm{c}}\nu_\lambda^{t-t_\mathrm{c}}\right]^{-1}.
\end{equation}
Each term $\nu_\lambda^{t-t_\mathrm{c}}$ in the sum is decreasing in $\lambda$ for $\lambda \in (0,1]$. Therefore, the right hand side of \cref{eq: final inequality proof two segment} is an increasing function in $\lambda$. This establishes the inequality of \cref{eq: final inequality proof two segment} for $t_\mathrm{c}>0$.

We conclude that $p_\mathrm{c}^*$ indeed satisfies $A_\lambda(t_\mathrm{c})\geq B_\lambda(p_\mathrm{c}^*)$. Hence, in the three-node chain, if the cutoff time $t_\mathrm{c}$ and probability $p_\mathrm{c}$ are such that the rates satisfy $R(t_\mathrm{c})=R(p_\mathrm{c})$, then $F(t_\mathrm{c})\geq F(p_\mathrm{c})$. 

The proof for $n_\mathrm{node}=3$ explicitly uses the analytic expressions of $R$ and $\overline{w}$. For $n_\mathrm{node}=4$ and $n_\mathrm{node}=5$ we do not have such analytic results, and so we resort to numerical evaluations of the rate-fidelity curves as in \cref{fig: rate fidelity curve} to support the claim that deterministic cutoffs yield higher fidelities than probabilistic cutoffs when compared at equal rate. Two examples for $n_\mathrm{node}=4$ and $n_\mathrm{node}=5$ are shown in \cref{fig: more rate fidelity curves}. More such examples may be generated using the code provided at \cite{Grimbergen2026}. For all parameter configurations that we have tried the claim holds.

\begin{figure*}
     \centering
     \begin{subfigure}[b]{0.45\textwidth}
         \centering
         \includegraphics[width=\textwidth]{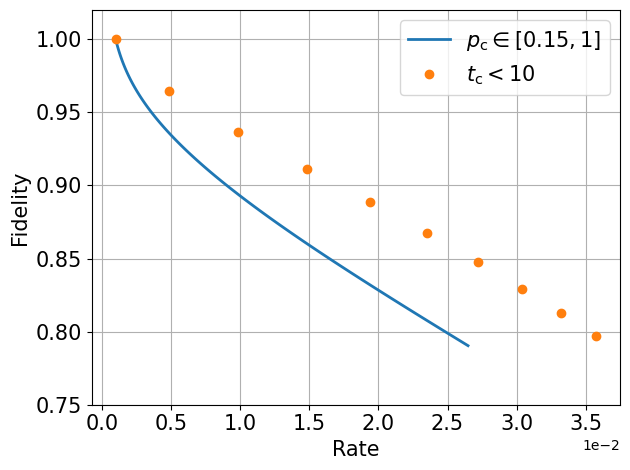}
         \caption{$n_\mathrm{node} = 4$}
     \end{subfigure}
     \hfill
     \begin{subfigure}[b]{0.45\textwidth}
         \centering
         \includegraphics[width=\textwidth]{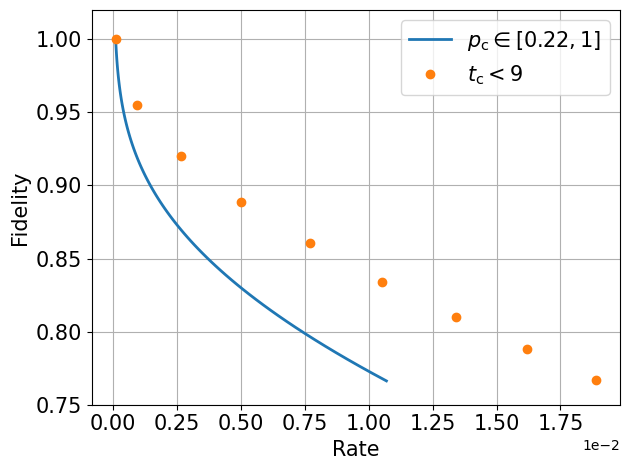}
         \caption{$n_\mathrm{node} = 5$}
     \end{subfigure}
    \caption{\textbf{Deterministic cutoffs provide higher fidelities than probabilistic cutoffs in chains with $n_\mathrm{node}=4$ and $n_\mathrm{node}=5$. }Rate-fidelity curves for quantum repeater chains with probabilistic or deterministic cutoff policy with (a) $n_\mathrm{node}=4$ and (b) $n_\mathrm{node}=5$. Notice that the scale of the rate-axes and the range of cutoff parameters shown are different between (a) and (b), the parameter ranges of the cutoff probability and age are adjusted so that the fidelity is at least above $0.75$. Apart from $n_\mathrm{node}$ the parameters are the same as in \cref{fig: rate fidelity curve}, namely $p_\mathrm{g}=0.1$, $\tau_\mathrm{coh}=20$, $w_0=1$ and $p_\mathrm{s}=1$. }
    \label{fig: more rate fidelity curves}
\end{figure*}

\section{Supplementary material: Secret-key rate}\label{appendix: SKR}
In this appendix we explain how the secret-key rate is optimized over the cutoff parameters $t_\mathrm{c}$ and $p_\mathrm{c}$ in the deterministic and probabilistic cutoff policy, respectively. We then discuss the approximate convergence of the ratio of the maximized secret-key rate between the policies as a function of $p_\mathrm{g}$ and $\tau_\mathrm{coh}$. 
Finally, we explain how the secret-key rates in \cref{fig: skr n} have been computed using Monte Carlo simulation.

\subsection{Maximization of secret-key rate}
We explain below how the secret-key rate is maximized over the cutoff parameters $p_\mathrm{c}$ and $t_\mathrm{c}$ for a fixed set of hardware parameters $n_\mathrm{node}$, $p_\mathrm{g}$, $p_\mathrm{s}$, $\tau_\mathrm{coh}$ and $w_0$. 
Recall that the secret-key rate is computed from the rate and expected Werner parameter according to \cref{eq:skr definition} as
\begin{equation}
    \SKR= R \cdot \SKF(\overline{w}).
\end{equation}

\paragraph{Probabilistic cutoffs} The maximized secret-key rate in the probabilistic cutoff policy, $\max_{p_\mathrm{c}\in [0,1]}\SKR(p_\mathrm{c})$, is numerically approximated as follows. We find the maximum of $\SKR(p_\mathrm{c})$ over a uniform discretization of the interval $[0,1]$. This gives an estimate of the value of $p_\mathrm{c}$ that maximizes the secret-key rate, up to the grid spacing of the discretization. The grid spacing that we use is $10^{-2}$. When the rates and fidelities are computed using the Monte Carlo simulation (discussed in \cref{section: monte carlo}) then we use this as the estimate for the optimal value of the cutoff probability $p_\mathrm{c}$. When the rates and fidelities are computed using the exact methods of \cref{appendix: Markov chain models}, then we sharpen the estimate of the optimal value of the cutoff probability $p_\mathrm{c}$ by using golden section search \cite[Chapter 10.2]{pressnumerical} to a precision of $10^{-4}$.

We note that the golden section search cannot be applied immediately to the interval $[0,1]$. Golden section search finds the maximum of a function with a single extremum. In most parameter regimes the secret-key rate as a function of $p_\mathrm{c}$ indeed has a single extremum. It typically looks like in \cref{fig:skr_vs_pc_n4_pg01_ps1_taucoh20_w01}. The local maximum corresponds to the optimal trade-off between rate (which increases when $p_\mathrm{c}$ is decreased) and secret-key fraction (which decreases when $p_\mathrm{c}$ is decreased). But in some parameter regimes the secret-key rate can have a local minimum near $p_\mathrm{c}=1$ as well. An example of this behaviour is shown in \cref{fig:skr_vs_pc_n4_pg07_ps1_taucoh20_w01}. The local minimum near $p_\mathrm{c}=1$ can arise because as the expected Werner parameter decreases the secret-key fraction initially decreases very rapidly, but later only modestly. Initially, then, the increase in rate due to the decrease in $p_\mathrm{c}$ does not weigh up against the dramatic decrease in secret-key fraction, and on the whole the secret-key rate decreases. This happens in particular when the elementary link generation probability is large and coherence times are small. However, as $p_\mathrm{c}$ gets further away from $1$ the secret-key fraction decreases smoothly to $0$, while the rate increases smoothly. A single local maximum can then form at the optimal trade-off between rate and secret-key fraction, as is the case in \cref{fig:skr_vs_pc_n4_pg07_ps1_taucoh20_w01}. The initial grid search returns a value for $p_\mathrm{c}$ around this local maximum, if it exists. The golden section search then localizes the local maximum more sharply. 

\begin{figure}[h]
    \centering
\begin{subfigure}[b]{0.45\linewidth}
    \includegraphics[width=\textwidth]{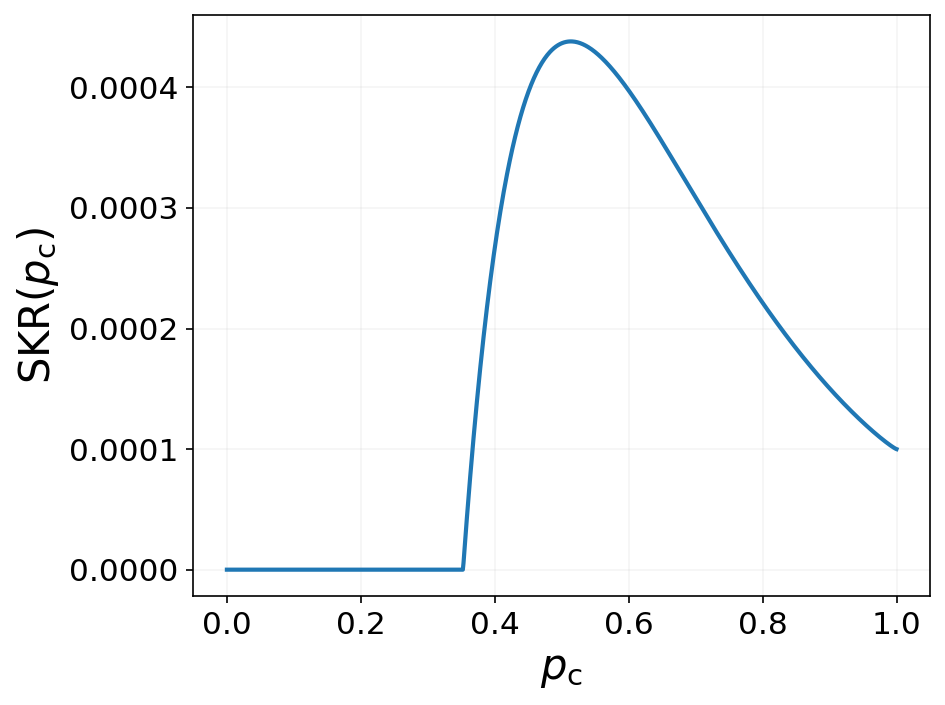}
    \caption{$p_\mathrm{g}=0.1$}
    \label{fig:skr_vs_pc_n4_pg01_ps1_taucoh20_w01}
\end{subfigure}
\begin{subfigure}[b]{0.45\linewidth}
    \includegraphics[width=\textwidth]{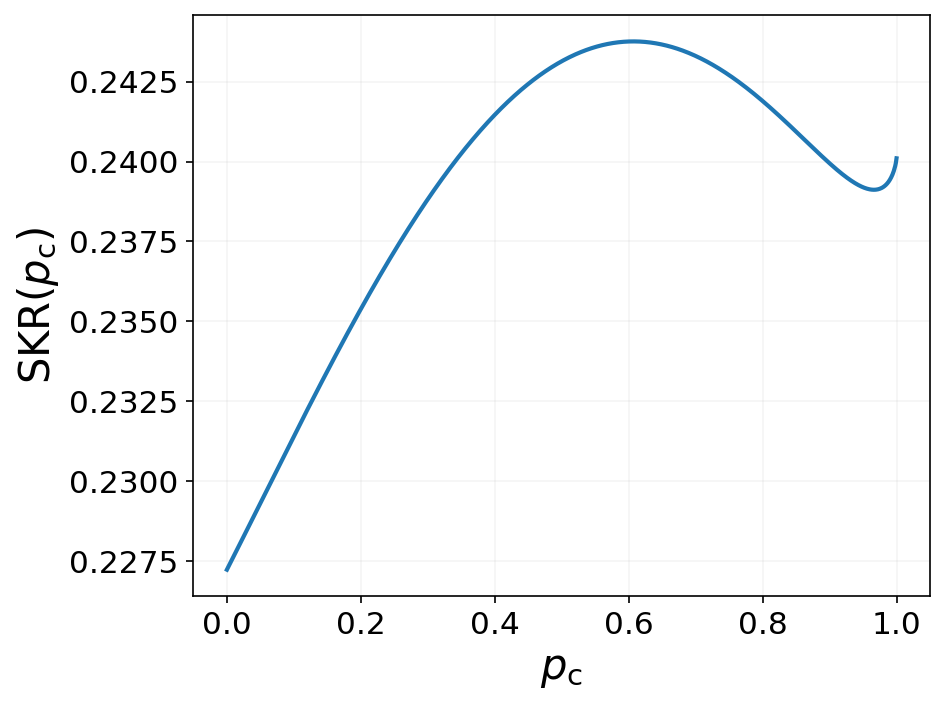}
    \caption{$p_\mathrm{g}=0.7$}
    \label{fig:skr_vs_pc_n4_pg07_ps1_taucoh20_w01}
\end{subfigure}
\caption{\textbf{The secret-key rate as function of $p_\mathrm{c}$ typically has a single local maximum, but in some cases it can also exhibit a local minimum near $p_\mathrm{c}=1$.} (a) The secret-key rate as function of $p_\mathrm{c}\in [0,1]$ for a five-node chain with parameters $p_\mathrm{g}=0.1$, $p_\mathrm{s}=1$, $\tau_\mathrm{coh}=20$ and $w_0=1$. The secret-key rate has a single local maximum as a function of $p_\mathrm{c}$. This is the case in most parameter configurations. (b) The secret-key rate as function of $p_\mathrm{c}\in [0,1]$ for a five-node chain with the same parameters as in (a), except that $p_\mathrm{g}=0.7$. In this case the secret-key rate has a local minimum near $p_\mathrm{c}=1$. Such a local minimum can form when the link generation probability $p_\mathrm{g}$ is large and the coherence time $\tau_\mathrm{coh}$ is small due the non-linearity of the secret-key fraction.}
\end{figure}

\paragraph{Deterministic cutoffs}
We now turn to the computation of the maximum secret-key rate in the deterministic cutoff policy, $\max_{t_\mathrm{c}\in \{0,1,2,\dots\}}\SKR(t_\mathrm{c})$. The optimal value of the cutoff time can in principle take any nonnegative integer value. However, for three-node chains we can prove the following two \textit{monotonicity properties}. Namely, that for all $t_\mathrm{c}\geq 0$,
\begin{enumerate}
    \item $R(t_\mathrm{c})\leq \lim_{t_\mathrm{c}\to\infty}R(t_\mathrm{c})$,
    \item $\overline{w}(t_\mathrm{c})\geq \overline{w}(t_\mathrm{c}+1)$.
\end{enumerate}
The first monotonicity property states that the rate with cutoffs is always less than the rate without cutoffs. Note that the rate without cutoffs can also be computed from the probabilistic cutoff policy by setting the cutoff probability to zero, that is, $\lim_{t_\mathrm{c}\to \infty}R(t_\mathrm{c})=R(p_\mathrm{c}=0)$.
The second monotonicity property states that the expected Werner parameter decreases as the cutoff time is increased.

Using these monotonicity properties we can guarantee that the optimal value of $t_\mathrm{c}$ is achieved below some finite value $t_\mathrm{c}^\mathrm{max}$. Indeed, they imply that the maximum secret-key rate that can be achieved for $t_\mathrm{c}\geq t_\mathrm{c}^\mathrm{max}$ is bounded by
\begin{equation}
    \max_{t_\mathrm{c}\geq t_\mathrm{c}^\mathrm{max}}\SKR(t_\mathrm{c}) \leq \left(\lim_{t_\mathrm{c}\to \infty}R(t_\mathrm{c})\right)\SKF(\overline{w}(t_\mathrm{c}^\mathrm{max})).
\end{equation}
Hence, we can guarantee that $\max_{t_\mathrm{c}\geq 0}\SKR(t_\mathrm{c})$ is found for $t_\mathrm{c}<t_\mathrm{c}^\mathrm{max}$ by verifying that
\begin{equation}
    \max_{t_\mathrm{c}<t_\mathrm{c}^\mathrm{max}}\SKR(t_\mathrm{c}) > \left(\lim_{t_\mathrm{c}\to \infty}R(t_\mathrm{c})\right) \SKF(\overline{w}(t_\mathrm{c}^\mathrm{max})).
\end{equation}

The proof of the two monotonicity properties for $n_\mathrm{node}=3$ uses the analytic formulas for the rate and expected Werner parameter given by \cref{eq: analytic three node deterministic appendix} and \cref{eq: analytic three node probabilistic appendix}. To show the first monotonicity property we note that $A_1(t_\mathrm{c})$ increases monotonically with $t_\mathrm{c}$ and hence $R=p_\mathrm{g}^2p_\mathrm{s}\flatfrac{(1+2A_1(t_\mathrm{c}))}{(1+2p_\mathrm{g}A_1(t_\mathrm{c}))}$ increases monotonically with $t_\mathrm{c}$. It follows immediately that $R(t_\mathrm{c})\leq \lim_{t_\mathrm{c}\to\infty}R(t_\mathrm{c})$ for three-node chains. For the second monotonicity property, we compute
\begin{align*}
    \overline{w}(t_\mathrm{c}+1) &= w_0^2 \frac{1+2A_\lambda(t_\mathrm{c}+1)}{1+2A_1(t_\mathrm{c}+1)} \\
    &= w_0^2\frac{1+2A_\lambda(t_\mathrm{c})}{1+2A_1(t_\mathrm{c}+1)} +w_0^2\frac{2(1-p_\mathrm{g})^{t_\mathrm{c}+1}\lambda^{t_\mathrm{c}+1}}{1+2A_1(t_\mathrm{c}+1)}
    \\
    &= \frac{1+2A_1(t_\mathrm{c})}{1+2A_1(t_\mathrm{c}+1)}\overline{w}(t_\mathrm{c})\\
    &\quad+\frac{(1+2A_1(t_\mathrm{c}+1))-(1+2A_1(t_\mathrm{c}))}{1+2A_1(t_\mathrm{c}+1)}w_0^2\lambda^{t_\mathrm{c}+1},
\end{align*}
where in the first step we use the definition of $A_\lambda(t_\mathrm{c})$ given in \cref{eq: A lambda sum} in terms of a sum over $t$. In the second step we use the expression of $\overline{w}(t_\mathrm{c})$ in terms of $A_\lambda(t_\mathrm{c})$ given by \cref{eq: analytic three node deterministic appendix} to get the first term and the definition of $A_1(t_\mathrm{c})$ in terms of the sum in \cref{eq: A lambda sum} to get the second term.
We have thus written $\overline{w}(t_\mathrm{c}+1)$ as a convex combination of $\overline{w}(t_\mathrm{c})$ and $w_0^2\lambda^{t_\mathrm{c}+1}$. Since $\overline{w}(t_\mathrm{c})$ may be written as $\overline{w}(t_\mathrm{c})=\sum_{t=0}^{t_\mathrm{c}}p_t w_0^2 \lambda^{t}$, where $p_t$ is the probability that the end-to-end link has age $t$, and $\lambda \leq 1$ by definition, it follows that $w_0^2\lambda^{t_\mathrm{c}+1}\leq \overline{w}(t_\mathrm{c})$. Hence, $\overline{w}(t_\mathrm{c}+1)\leq \overline{w}(t_\mathrm{c})$. This proves the second monotonicity property for three-node chains.

We have not been able to find a general proof that the monotonicity properties hold for longer chains. We therefore treat them as assumptions for chains longer than three nodes. We have numerically verified that the two monotonicity properties are at least valid up to $t_\mathrm{c}^\mathrm{max}=30$ for $n_\mathrm{node}=4$ and $t_\mathrm{c}^\mathrm{max}=20$ for $n_\mathrm{node}=5$, in all cases considered.

\subsection{Convergence of ratios of maximized secret-key rates}
\cref{fig: skr pg} in the main text shows the maximized secret-key rate as function of $p_\mathrm{g}\in [10^{-3},1]$ for $n_\mathrm{node}=3$, $4$ and $5$ at $\tau_\mathrm{coh}=50$. The ratios $\max_{p_\mathrm{c}}\SKR(p_\mathrm{c})/\max_{t_\mathrm{c}}\SKR(t_\mathrm{c})$ of these maximized secret-key rates are shown in \cref{fig: skr pg ratios}. It can be seen that the ratios exhibit convergence as $p_\mathrm{g}$ approaches $10^{-3}$ from above. 

In \cref{fig: skr tau_coh ratios} we show the ratio of maximized secret-key rates at $p_\mathrm{g}=10^{-3}$ as a function of $\tau_\mathrm{coh}$. In this case, the ratios show convergence around $\tau_\mathrm{coh}=50$. Hence, as stated in the main text, the ratios of the maximized secret-key rates at $p_\mathrm{g}=10^{-3}$ and $\tau_\mathrm{coh}=50$ give an indication of the worst-case secret-key rate loss of the probabilistic cutoff policy with respect to the deterministic cutoff policy, in the parameter regime explored.

\begin{figure}
    \centering
    \includegraphics[width=0.9\linewidth]{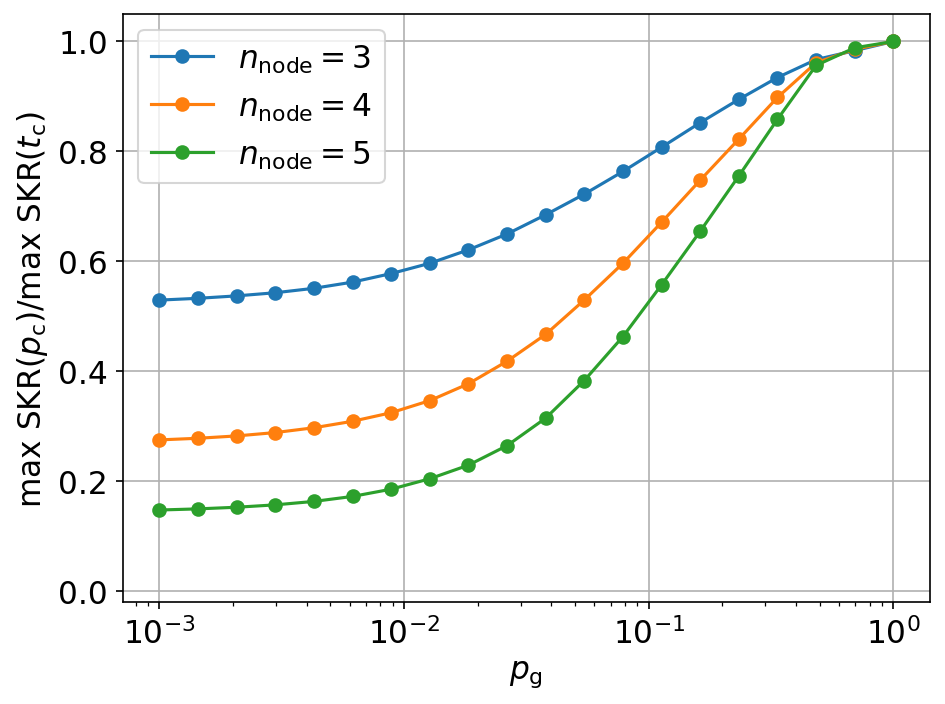}
    \caption{Ratios $\flatfrac{\max_{p_\mathrm{c}}\SKR}{\max_{t_\mathrm{c}}\SKR}$ of maximized secret-key rates in \cref{fig: skr pg} for $p_\mathrm{s}=1$, $\tau_\mathrm{coh}=50$ and $w_0=1$ show convergence around $p_\mathrm{g}=10^{-3}$.}
    \label{fig: skr pg ratios}
\end{figure}

\begin{figure}
    \centering
    \includegraphics[width=0.9\linewidth]{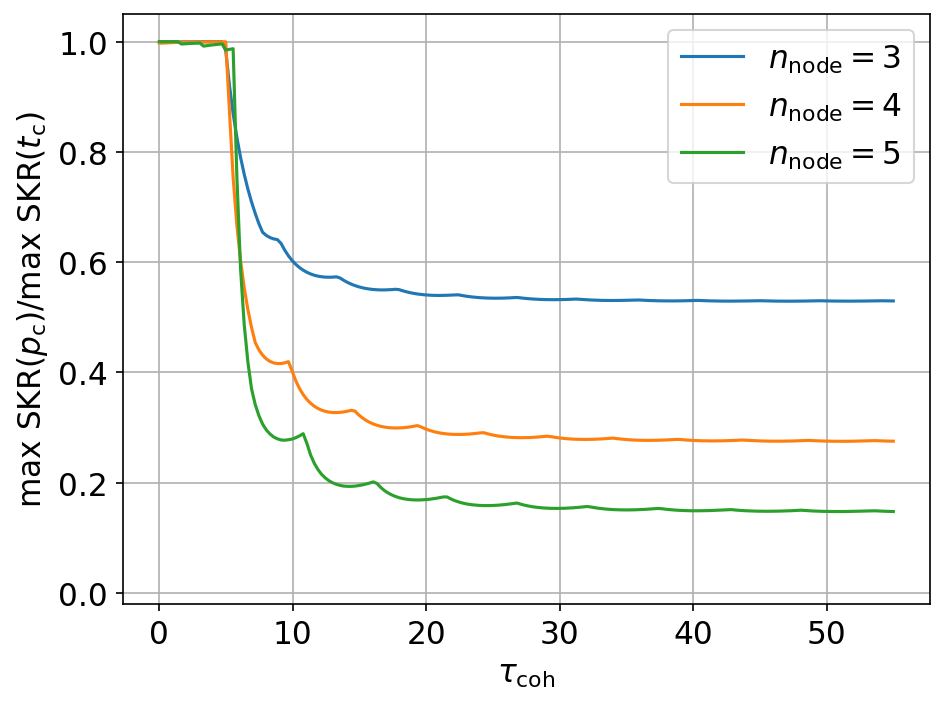}
    \caption{Ratios $\flatfrac{\max_{p_\mathrm{c}}\SKR}{\max_{t_\mathrm{c}}\SKR}$ of maximized secret-key rates at $p_\mathrm{g}=10^{-3}$, $p_\mathrm{s}=1$ and $w_0=1$ as function of $\tau_\mathrm{coh}$ show convergence around $\tau_\mathrm{coh}=50$.}
    \label{fig: skr tau_coh ratios}
\end{figure}

\subsection{Monte Carlo simulation}\label{section: monte carlo}
Complementary to the Markov methods of \cref{appendix: Markov chain models}, the rate and expected Werner parameter can also be approximated through a Monte Carlo simulation. One time step of the Monte Carlo simulation corresponds to one time step of the repeater chain. The number of steps required to generate a single sample of end-to-end link generation thus requires $T$ time steps, where $T$ is the delivery time of the end-to-end link. The expected number of time steps needed to generate one sample of an end-to-end link is $\overline{T}$. Therefore, the method can only be reasonably used for high values of the elementary link generation probability $p_\mathrm{g}$. We use the Monte Carlo simulation to approximate the rate and expected Werner parameter for chains with a number of nodes beyond the reach of our exact methods. 

A single sample $T^{(i)}$ of the end-to-end delivery time $T$ and a single sample $w^{(i)}$ of the end-to-end Werner parameter $w$ are obtained by simulating the evolution of the repeater chain from the empty chain until an end-to-end link is generated. Each HEG attempt is simulated by sampling a Bernoulli trial with success probability $p_\mathrm{g}$. Entanglement swaps and cutoffs are subsequently performed according to the swap-asap policy and the corresponding cutoff policy. The simulation ends when an end-to-end link is created. The number of time steps until the end-to-end link is generated is recorded as the sample $T^{(i)}$ and the Werner parameter of the end-to-end link is recorded as the sample $w^{(i)}$. After generating $N_\mathrm{samples}$, the expected delivery time $\overline{T}$ and expected Werner parameter $\overline{w}$ are approximated by the sample means
\begin{equation}
    \widehat T = \frac{1}{N_\mathrm{samples}}\sum_{i=1}^{N_\mathrm{samples}}T^{(i)}    
\end{equation}
and
\begin{equation}
\widehat w = \frac{1}{N_\mathrm{samples}}\sum_{i=1}^{N_\mathrm{samples}}w^{(i)}. 
\end{equation}
The code at \cite{Grimbergen2026} contains a verification of the simulated $\widehat T$ and $\widehat w$ against the exactly computed $\overline{T}$ and $\overline{w}$ from \cref{appendix: Markov chain models} for chains up to five nodes.

For the estimation of the secret-key rate we recall that the secret-key rate depends on the \emph{expected} delivery time and \emph{expected} Werner parameter. Given estimates $\widehat T$ and $\widehat w$ for these quantities an estimate for the secret-key rate can be obtained as
\begin{equation}
    \widehat\SKR = \widehat T^{-1} \SKF(\widehat w).
\end{equation}
This estimate is made in accordance with the definition of the rate as $R=\overline{T}^{-1}$. In this way, we can approximate the secret-key rate. But it is a non-trivial task to obtain the standard deviation on this estimate by error propagation of the uncertainty in $\overline{w}$ since the secret-key rate is a highly non-linear function of $\overline{w}$. Therefore, to obtain statistics on the secret-key rate, we compute multiple estimates of $\overline{T}$ and $\overline{w}$, so that we also get multiple estimates of $\SKR$. We do this by computing $\widehat{T}$ and $\widehat{w}$ for $N_\mathrm{batches}$ batches of $N_\mathrm{samples}$ samples each. This yields estimates $\widehat{T}^{(j)}$ of the expected delivery time and $\widehat{w}^{(j)}$ of expected Werner parameter, where the index $j=1,\dots,N_\mathrm{batches}$ runs over the batches. 
The secret-key rate is then approximated as the mean of $\widehat{\SKR}^{(j)}$ over $j=1,\dots, N_\mathrm{batches}$. The simulated secret-key rates shown in \cref{fig: skr n} are the result of maximizing over the mean secret-key rates obtained over $N_\mathrm{batches}=20$ batches of $N_\mathrm{samples}=100$ samples each. The error bars are the standard deviation in the $\widehat \SKR^{(j)}$ over $j=1,\dots, N_\mathrm{batches}$. 
The maximization over $p_\mathrm{c}$ is done by taking the maximum mean value for the secret-key rate obtained over a uniform discretization of the unit interval $p_\mathrm{c}\in [0,1]$ into steps of size $10^{-2}$. The maximization over $t_\mathrm{c}$ is done by taking the maximum mean value for the secret-key rate obtained for $t_\mathrm{c}\leq \tau_\mathrm{coh}$.  


\section{Deterministic cutoff policy with end-to-end cutoff}\label{appendix: other cutoff policies}
The deterministic cutoff policy that has been considered in the main text as a benchmark was selected for its similarity to the probabilistic cutoff policy. Other choices of deterministic cutoff policy are also possible. We discuss one of them in this appendix. 

A second deterministic cutoff policy can be obtained from the one that we have considered so far by adding one extra rule. The extra rule that we introduce is a post-selection rule, or end-to-end (e2e) cutoff: If the end-to-end link has age $t>t_\mathrm{c}$, then it is discarded. We call this the \textit{deterministic cutoff policy with e2e cutoff}. 

We note that for three-node chains there is only one swap needed to create the end-to-end link, and the assumption of swap-asap implies that one of the links involved in this swap will always have age $0$. Even without the e2e cutoff this already guarantees that the age of the end-to-end link will be less than $t_\mathrm{c}$. Therefore, for $n_\mathrm{node}=3$ the deterministic cutoff policy is not changed by the introduction of the e2e cutoff.

However, for $n_\mathrm{node}>3$ the e2e cutoff does yield a different deterministic cutoff policy. The fact that all internal links have age less than $t_\mathrm{c}$ does not guarantee anymore that this will also be true for the end-to-end link. We illustrate this with an example for the four-node chain with $t_\mathrm{c}=1$, displayed in \cref{fig:post-selection-pathway} and explained below. 

Suppose that in the first time step the two links on the edges of the repeater chain are generated, but the generation of the central link fails. Both links on the edges are kept because their age is less than $t_\mathrm{c}$. During the HEG round of the second time step the two edge links each age by one time step. If the central link is then generated and all swaps are performed, this results in an end-to-end link of age $2$ since the ages of all links involved in the swap have to be added. The age of the end-to-end link in this case exceeds the cutoff time, and has to be discarded. From this example we guess that the e2e cutoff will lower the rate, because some end-to-end links that were accepted without post-selection are now being discarded. On the other hand, it is expected to increase the fidelity since the end-to-end links that are accepted will be younger on average. These claims can be verified through exact computation of the rates and fidelities.

\begin{figure}
    \centering
    \includegraphics[width=0.9\linewidth]{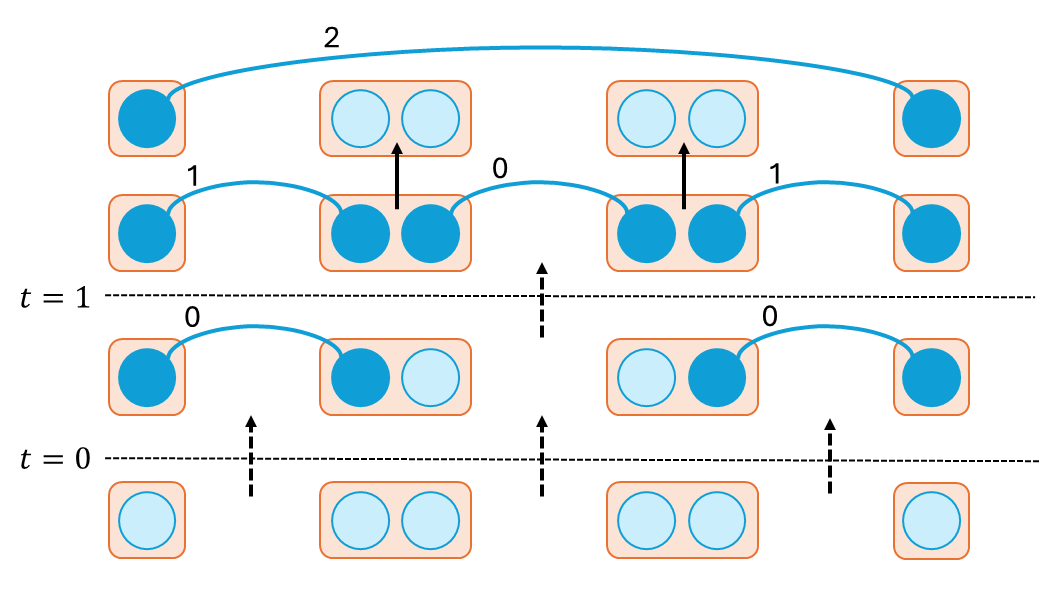}
    \caption{\textbf{Example of sequence of link generation in four-node chain that triggers the end-to-end cutoff when $t_\mathrm{c}=1$.} Link ages are displayed next to links. Dashed arrows indicate HEG attempts and solid arrows indicate entanglement swaps. As explained in the text, this sequence of link generations provides the intuition for why end-to-end cutoffs can hurt repeater chain performance at high elementary link generation probability, as observed in \cref{fig:large pg rate fidelity curves with post-selection} and \cref{fig:large pg rate fidelity curves with post-selection rate suppression}.}
    \label{fig:post-selection-pathway}
\end{figure}

The rates and fidelities for the deterministic cutoff policy with e2e cutoff can be computed using the same type of Markov chain model as described in \cref{appendix: Markov chain models} for the case without e2e cutoffs. The only difference is that transitions into end-to-end states with age $t>t_\mathrm{c}$ have to be replaced with transitions into the empty state. The resulting Markov chain has been implemented in the code at \cite{Grimbergen2026}.

From this Markov chain model we confirm that, at the same cutoff time $t_\mathrm{c}$, the e2e cutoff improves fidelity at the cost of a lower rate. This can be seen in \cref{fig: rate fidelity curves with post-selection}. For small values of $p_\mathrm{g}$ the decrease in rate and increase in fidelity due to the e2e cutoff lead to an improved rate-fidelity trade-off, as can be seen from the fact that the green squares are above the orange dots in \cref{fig:small pg rate fidelity curves with post-selection}. However, for large values of $p_\mathrm{g}$ the decrease in rate is so drastic that it can worsen the rate-fidelity trade-off. In the case shown in \cref{fig:large pg rate fidelity curves with post-selection}, at $p_\mathrm{g}=0.8$, we even see that the e2e cutoff causes the deterministic cutoff to yield lower rates and fidelities than the probabilistic cutoff policy.

In fact, for $p_\mathrm{g}=0.85$ it can be seen in \cref{fig:large pg rate fidelity curves with post-selection rate suppression} that the rate at $t_\mathrm{c}=1$ is suppressed so much by the post-selection, that it is even less than the rate at $t_\mathrm{c}=0$. To see why this can happen we return to the example of \cref{fig:post-selection-pathway}. 

If for $t_\mathrm{c}=1$ the two links on the edge of a four-node chain are generated, then they must be stored. But in the next time step, independent of whether or not the central link will be generated, there will not be an acceptable end-to-end link. If the central link is not generated, then the links on the edges are discarded because they violate the cutoff condition. On the other hand, if the central link is generated, then the end-to-end link itself violates the post-selection condition and it will also be discarded. Hence, storing the two edge links will invariably waste the next time step if $t_\mathrm{c}=1$. On the other hand, for $t_\mathrm{c}=0$ this time step is not wasted because the chain returns to the empty state immediately after having generated the two edge links. Our computations show that for large enough $p_\mathrm{g}$ this wasting of time steps can decrease the rate at $t_\mathrm{c}=1$ below the rate of $t_\mathrm{c}=0$. 

This analysis of the deterministic cutoff policy with e2e cutoff shows that strict constraints on the fidelity that improve entanglement distribution performance in one parameter regime may actually be harmful in other regimes. One could work around this problem by using more complicated deterministic cutoff policies with additional rules. In the four-node chain one could for example make sure to never store the two edge links when they cannot yield an acceptable end-to-end link.

\begin{figure*}[b]
     \centering
     \begin{subfigure}[b]{0.3\textwidth}
         \centering
         \includegraphics[width=\textwidth]{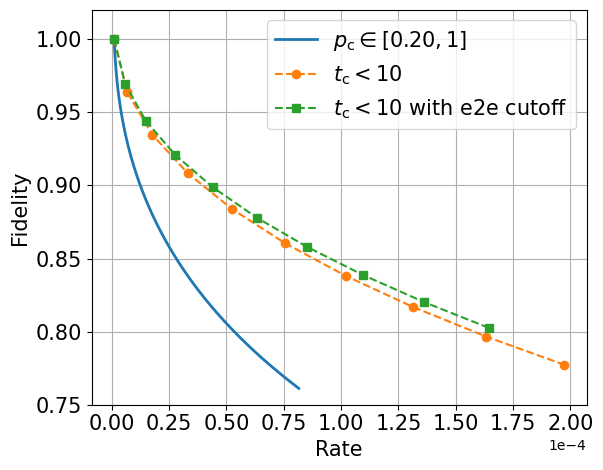}
         \caption{$p_\mathrm{g}=0.01$}
         \label{fig:small pg rate fidelity curves with post-selection}
     \end{subfigure}
     \hfill
     \begin{subfigure}[b]{0.3\textwidth}
         \centering
         \includegraphics[width=\textwidth]{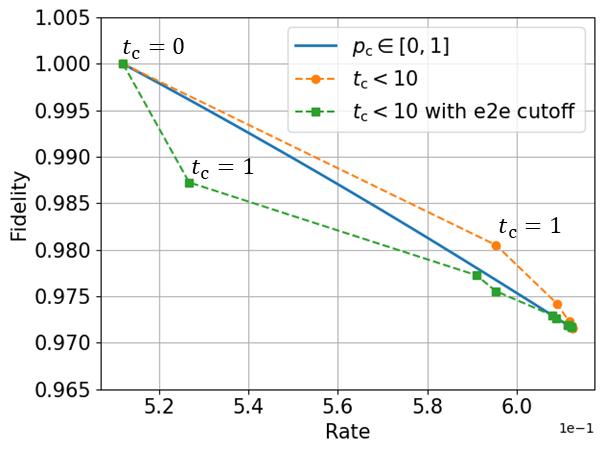}
         \caption{$p_\mathrm{g}=0.8$}
         \label{fig:large pg rate fidelity curves with post-selection}
     \end{subfigure}
     \hfill
     \begin{subfigure}[b]{0.3\textwidth}
         \centering
         \includegraphics[width=\textwidth]{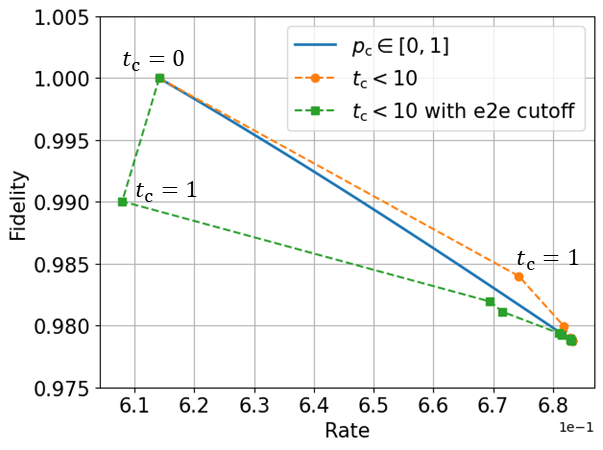}
         \caption{$p_\mathrm{g}=0.85$}
         \label{fig:large pg rate fidelity curves with post-selection rate suppression}
     \end{subfigure}
    \caption{\textbf{Including e2e cutoffs in the deterministic cutoff policy improves the fidelity but suppresses the rate that is achieved with a particular cutoff time.} Rate-fidelity curves for quantum repeater chains with probabilistic cutoff policy ($p_\mathrm{c}$), deterministic cutoff policy ($t_\mathrm{c}$) and deterministic cutoff policy with e2e cutoff ($t_\mathrm{c}$ with e2e cutoff). Results shown are for $n_\mathrm{node}=4$ and $\tau_\mathrm{coh}=20$ and have been obtained from exact Markov chain computations as described in \cref{appendix: Markov chain models}. The discrete markers used for the deterministic cutoff policies have been connected with dashed lines to visualize the order of the cutoff times. (a) For small $p_\mathrm{g}$ the e2e cutoff improves the rate-fidelity trade-off of the deterministic cutoff policy. More precisely, comparing the two deterministic cutoff policies at the same cutoff time $t_\mathrm{c}$, we see that the rate is slightly lowered by the e2e cutoff but the fidelity is increased. (b) For large $p_\mathrm{g}$ the fidelity is still improved by the e2e cutoff, but the decrease in rate can be drastic. In fact, for $p_\mathrm{g}$ around $0.8$ the rate is decreased so much that even the probabilistic cutoff policy achieves higher rates and fidelities. (c) For $p_\mathrm{g}=0.85$ the rate with e2e cutoffs at $t_\mathrm{c}=1$ is even suppressed below the rate at $t_\mathrm{c}=0$.}
    \label{fig: rate fidelity curves with post-selection}
\end{figure*}

\end{appendices}

\end{document}